\documentclass[12pt,draftcls,onecolumn]{IEEEtran}
\usepackage{amsmath}
\usepackage{amsthm}
\usepackage{amsfonts}
\usepackage{amssymb}
\usepackage[utf8]{inputenc}
\DeclareUnicodeCharacter{00A0}{ }
\usepackage{graphicx}
\usepackage{graphics}
\usepackage{float}
\usepackage{tikz}
\usetikzlibrary{shapes,snakes,patterns}
\usepackage{epstopdf}
\usepackage{booktabs}

\usepackage[justification=centering]{caption}
\usepackage{enumerate} 
\usepackage{cite}
\usepackage{placeins}
\newcommand{\vast}{\bBigg@{4}}
\newcommand{\Vast}{\bBigg@{5}}
\newtheorem{theorem}{Theorem}[]
\newtheorem{lemma}[]{Lemma}
\usepackage{suffix}
\usepackage{mathtools}
\usepackage{amsthm}
\usepackage{xfrac}
\usepackage{relsize}
\usepackage{xcolor}
\usepackage[inline]{enumitem}
\theoremstyle{definition}
\newtheorem{defn}{Definition}[]

\usepackage[list=true]{subcaption}
\usepackage{pifont}
\usepackage{changebar}
\newcommand{\cmark}{\ding{51}}%
\newcommand{\xmark}{\ding{55}}%
\title{Transmit Power Policy and Ergodic Multicast Rate Analysis of Cognitive Radio Networks in Generalized Fading}
\author{Athira Subhash$^{*}$, Muralikrishnan Srinivasan$^{*}$, Sheetal Kalyani, Lajos Hanzo {\footnote{{Athira Subhash, Muralikrishnan Srinivasan and Sheetal Kalyani are with the Department of Electrical Engineering, Indian Institute of Technology Madras. (email:\{ee14d206@ee,ee16d027@smail,skalyani@ee\}.iitm.ac.in).\\Lajos Hanzo is with School of Electronics and Computer Science, University of Southampton.(email:hanzo@soton.ac.uk).\\$^*$Athira Subhash and Muralikrishnan Srinivasan are co-first authors.}}} } 

\begin{document}
	
	\maketitle
	\begin{abstract}
		This paper determines the optimum secondary user power allocation and ergodic multicast rate of point-to-multipoint communication in a cognitive  radio  network in the presence of outage constraints for the primary users. 
		Using tools from extreme value theory (EVT), it is first proved that the limiting distribution of the minimum of independent and identically distributed (i.i.d.) signal-to-interference ratio (SIR) random variables (RVs) is a Weibull distribution, when the user signal and the interferer signals undergo independent and non-identically distributed (i.n.i.d.) $\kappa-\mu$ shadowed fading. Also, the rate of convergence
		of the actual minimum distribution to the Weibull distribution is derived. This limiting distribution is then used for determining the optimum transmit power of a secondary network in an underlay cognitive radio network subject to outage constraints at the primary network in a generalized fading scenario. Furthermore, the asymptotic ergodic multicast rate of secondary users is analyzed for varying channel fading parameters. 
	\end{abstract}
	\begin{IEEEkeywords}
		extreme value theory, $\kappa-\mu$ shadowed fading, outage probability, cognitive radio
	\end{IEEEkeywords}
	\section{Introduction}
	\par With the advances in wireless technology, the presence of wireless devices has become ubiquitous. Furthermore, with the advent of the Internet of Things (IoT), the number of connected devices accessing the spectrum is set to increase in the coming times. With the upcoming increase in devices and hence increasing the traffic, it will be very hard to find free spectrum. Cognitive radio (CR) is one of the promising techniques mitigating spectrum scarcity in wireless communication systems \cite{dalta2009cr, wang2011cr, liang2011cr, goldsmith2009cr,miridakis2018mimo}. In cognitive radio networks (CRNs), there are three popular modes of spectrum sharing between primary users (PU) and secondary users (SU) - underlay, overlay and interweave \cite{zhao2007cr,khoshkholgh2010cr,patel2017achievable,patel2018many}. As a further development, the authors of \cite{zou2015relay,ding2019security} have studied the security aspects of a CR system in the presence of eavesdroppers.  
	
	\par Throughout this paper, we consider the underlay mode, in which the secondary network is allowed to access the spectrum allocated to the primary network provided that the interference caused by the SU transmitter does not unduly deteriorate the performance of the primary network. An important problem in CRNs is the choice of power policy at the SU-Tx (transmitter), so that the interference at the PU-Rx (receiver) remains below an admissible threshold. Several authors \cite{zhang2009cr, suraweera2010cr, rezki2012cr, kang2011cr, smith2013cr} have studied the performance of underlay CRNs under various interference constraints. In \cite{hanif2017cr}, different-power adaptive transmit antenna selection (TAS) schemes were analyzed for the underlay CRN. Furthermore, the authors of \cite{patel2016achievable} have determined the optimal rate sharing parameters for both the SU and the PU, so that the achievable rates were maximized. Similarly, recent contributions \cite{louis_ergodic, badarneh_asymptotic} have also considered the performance of an interference-limited underlay CRN relying on continuous power adaptation at the SU. In \cite{badarneh_asymptotic}, the secondary transmitter is assumed to transmit information to the specific SU, having the $k$th highest signal-to-interference ratio (SIR). The authors of \cite{louis_ergodic} have also investigated the ergodic capacity of the secondary network in an underlay CRNs contaminated by the interference arriving from the primary network in conjunction with various scheduling schemes, including a multicast scheduling (MS) scheme designed for enhancing the fairness among the users. The authors of \cite{miridakis2018mimo} studies the optimal power allocation, the effective number of secondary transmit antennas, the efficient trade off between transmit-and-harvest secondary antennas, and the average  channel  capacity  of  an energy harvesting-enabled  secondary  system in a massive MIMO CRN. In Table \ref{literture}, we provide a bold summary and comparison of the seminal literature relying on system models similar to our scenario. 
	\par The authors of \cite{louis_ergodic, badarneh_asymptotic} consider the analysis of power policy at the SU and the ergodic capacity of the SU in Rayleigh fading channels. Our focus in this treatise is on extending these results to general fading scenarios. At the time of writing, generalized multipath fading models such as the $\kappa-\mu$ and the $\eta-\mu$ fading distributions are generating significant research interests \cite{yacoub_k_mu}. They model the small-scale variations in the fading channel in line of sight (LOS) and non-line of sight (NLOS) conditions respectively. To investigate the effects of shadowing on the dominant LOS component, the authors of \cite{paris2014statistical} and \cite{cotton_d2d} have developed a generalization of the shadowed Rician fading called the $\kappa-\mu$ shadowed fading model. The $\kappa-\mu$ shadowed fading has been shown to unify the $\kappa-\mu$ and $\eta-\mu$ fading models \cite{pozas_shadowed} and to have a wide variety of applications ranging from land-mobile satellite systems to device-to-device communication \cite{cotton_d2d}. Performance metrics conceived for generalized fading have been studied extensively in \cite{celia2014capacity, zhang2015effective, chen2016outage, li2017rate, zhang2017hos, chandrasekaran2015performance,thomas2016error, morales2012outage, paris2013outage, ermolova2014outage, kumar2015coverage, kumar2015outage,zhang2017performance, parthasarathy2017coverage, parthasarathy2018evm, srinivasan2018secrecy}. The exact outage and rate expressions in the presence of co-channel interference (CCI) were studied in \cite{kumar2017outage} only quite recently. 
	\begin{table}[H]
		\centering
		\begin{tabular}{|c|c|c|c|c|c|c|c|}
			\hline
			& Our model &\cite{badarneh_asymptotic}-2019 &\cite{aghazadeh2018performance}-2018  & \cite{louis_ergodic}-2016 & \cite{khan2015performance}-2015& \cite{kang2011cr}-2011     \\ \hline
			
			Number of PU-Rx & Multiple & Single & Single &Multiple  & Multiple & Single  \\ \hline
			
			Number of SU-Rx & Multiple & Multiple & Single  & Multiple & Single & Single  \\ \hline
			
			\textbf{\begin{tabular}[c]{@{}c@{}}Instantaneous(I)/\\ Statistical(S) CSI\end{tabular}} & S & \begin{tabular}[c]{@{}c@{}}Su-Tx to Pu-Rx : I\\ Rest of the links : S\end{tabular} & \textcolor{black}{both} & S & I & I    \\ \hline
			
			Channel Fading &$\kappa-\mu$ shadowed &Nakagami & Rayleigh  & Rayleigh & \begin{tabular}[c]{@{}c@{}} Nakagami\end{tabular}  & -  \\ \hline
			
			\textbf{\begin{tabular}[c]{@{}c@{}}Interference from\\ PU-Tx\end{tabular}} & \cmark & \cmark & \xmark & \cmark & \xmark &  \cmark   \\ \hline
			
			\textbf{{\begin{tabular}[c]{@{}c@{}}Outage constraints \\ for PU-Rx\end{tabular}}} & \cmark &  \xmark & \xmark & \cmark & \xmark & \cmark   \\ \hline
			
			Usage of EVT & \cmark & \cmark & \xmark & \xmark & \xmark & \xmark   \\ \hline
			\textbf{{\begin{tabular}[c]{@{}c@{}}Expression for  \\ secondary capacity\end{tabular}}} & \cmark & \xmark & \xmark & \cmark & \xmark & \xmark   \\ \hline
		\end{tabular}
		\caption{Comparison with existing literature.}
		\label{literture}
	\end{table}
	
	\par A feature that is common among the above contributions is the complicated nature of the PDF and the CDF of the SIR \cite{chen2016outage}, \cite{morales2012outage, paris2013outage, ermolova2014outage, kumar2015coverage, kumar2015outage, parthasarathy2017coverage, parthasarathy2018evm, kumar2017outage}. For example, the recent work \cite{kumar2017outage}, which generalizes all existing results, considers the scenario when the signal of interest (SOI) and the CCI to undergo i.n.i.d. $\kappa-\mu$ shadowed fading and derives the CDF of SIR in terms of an infinite summation of the Lauricella function of the fourth kind. This Lauricella function itself involves an N-fold infinite summation, where $N$ denotes the number of interferers. The complementary cumulative distribution function (CCDF) of the minimum of independent random variables (RVs) is given by the product of the CCDF of each of the variables. Hence, in the case of $L$ independent and identically distributed (i.i.d.) RVs, the CCDF of the minimum is given by the $L$th power of the common CCDF. 
	\par Determining the CDF/CCDF of the minimum of $L$ such i.i.d. SIR realizations has a direct application in SU power control and in deriving the ergodic multicast rate in CRN \cite{louis_ergodic, badarneh_asymptotic}. Calculating the outage constraints over several PU-Rx requires the knowledge of the CDF of the minimum SIR. Furthermore, the ergodic rate of the MS scheme in the secondary network is determined directly by the SIR of the weakest user. However, the need for raising the CCDF of SIR random variables to power $L$ makes the corresponding mathematical analysis very difficult. In fact, even the evaluation of the exact CCDF of the minimum of two SIR RVs with each SIR RV having two i.n.i.d interferers in a $\kappa-\mu$ shadowed fading environments takes more than an hour to compute in Mathematica when one uses the series expansion of \cite[Eqn. (8)]{kumar2017outage}. Further more, the evaluation of the exact CCDF of the minimum of four SIR RVs with each SIR RV having four i.n.i.d interferers in a $\kappa-\mu$ shadowed fading environment times out in Mathematica. Hence, it is imperative that a simple limiting distribution is found for the minimum.
	
	\par Extreme Value Theory (EVT) has been routinely used in the literature for characterizing the asymptotic maximum or minimum SIR in terms of very simple probability distribution functions (PDF)/CDFs that are amenable to analysis \cite{jindal2006multicast, park2008multicast, park2009multicast, oyman2007scheduling,oyman2008scheduling, oyman2010scheduling,ahmadi2012scheduling, kountouris2009scheduling, xue2010mi, xia2014scheduling, biswas2016relay, xu2016relay, kalyani2012gamma, pun2011mimo}. Quite recently, the authors of \cite{gao2018massive} derived the statistical upper channel capacity bounds for FAS systems using EVT in the large-scale limit for Rayleigh fading channels.  Furthermore, in \cite{badarneh_asymptotic} EVT is used for determining the average throughput of the k-th best SU under continuous power adaptation at the SU. {Even though these are asymptotic results, they are observed to hold fairly well even for $20$ receivers in the SU network. } In an interference-limited scenario when the source and interferers undergo i.n.i.d. $\kappa-\mu$ shadowed fading, the authors of \cite{subhash2019asymptotic} use EVT for proving that the limiting distribution of the maximum of SIR RVs converges to a Frechet distribution \cite{gumbel2012statistics} and further derives the corresponding rate of convergence. They also prove the convergence of moments of the true maximum distribution to the moments of the asymptotic maximum distribution.
	\par Against this backdrop, in this work we use EVT to determine the power adaptation at the SU underlay in an CRN,  subject to specific outage constraints for the primary users. We also use EVT for determining the ergodic multicast rate of the SUs. Our main contributions in this paper are as follows:
	\begin{itemize}
		\item Assuming that the user signal and the interferer signal undergo i.n.i.d. $\kappa-\mu$ shadowed fading, we prove that the limiting distribution of the minimum of $L$ such i.i.d. SIR RVs is a Weibull distribution. 
		\item We also derive the rate of convergence of the actual distribution of the minimum SIR to the derived asymptotic distribution.
		\item Using the limiting distribution derived, we determine a closed form expression for the optimum power to be used at the SU-Tx while ensuring that the outage constraints at the PU-Rx are met. 
		\item Further, we derive expressions for the ergodic multicast rate of point-to-multipoint communications in the secondary network.
	\end{itemize} 
	Note that the above mentioned  results hold for Rayleigh, Rician, Nakagami-m, $\kappa-\mu$ and $\eta-\mu$ faded user and interferer fading scenarios since all of these are special cases of $\kappa-\mu$ shadowed fading. Since we assume i.n.i.d. interferers, we also account for interferers having different path-loss or having unequal powers.

	\section{System model} \label{Application}
	We consider a CR scenario where the PU network consists of a PU‐Tx serving $M$ multicast PU-Rxs and a SU network that consists of a SU-Tx serving $L$ multicast SU‐Rxs. Here, all the devices have a single antenna for transmission/reception.  Furthermore, here we assume that the SU‐Tx sends common multicast information to all the SU-Rxs in the underlay mode. Since an underlay mode is considered, the SU-Tx has to rely on continuous power adaptation strategy for satisfying the instantaneous interference constraints at the PU-Rxs. The channel power gains of the links PU‐Tx $\to$ PU‐Rx$_m$, for $m = 1,2,…,M$ and SU‐Tx $\to$ SU‐Rx$_l$, for $l = 1,2,…,L$ are denoted by $h_m$, for $m = 1,2,…,M$ and $g_l$, for $l = 1,2,…,L$, respectively. Similarly, $\alpha_m$ and $\beta_l$ are the channel power gains of the interference links SU‐Tx$\to$PU‐Rx$_m$ and PU‐Tx $\to$ SU‐Rx$_l$, respectively. All the channels are considered to undergo $\kappa-\mu$ shadowed fading. Furthermore, we consider an interference-limited system, where the noise power at each of the SU-Rx (or PU-Rx) is negligible compared to the interference power received from the PU-Tx (or SU-Tx). The authors of \cite{louis_ergodic, badarneh_asymptotic} consider a similar system model except for the fact that they assume Rayleigh faded channels. Furthermore, the authors of \cite{badarneh_asymptotic} consider only one PU-Rx. The instantaneous SIRs at the $m$th PU‐Rx and $l$th SU‐Rx are
	\begin{equation} \label{sir_p}
	\gamma_{m, p}=\frac{P_p h_m}{P_s \alpha_m}, \quad m=1,..., M,
	\end{equation}
	and 
	\begin{equation}\label{sir_s}
	\gamma_{l,s}=\frac{P_s g_l}{P_p \beta_l}, \quad l=1,...,L,
	\end{equation}
	respectively.
	Here, $P_p$ is the PU‐Tx transmit power, $P_s$ is the instantaneous SU‐Tx  transmit power and $\{h_m,\alpha_m,g_l,\beta_l;m = 1,2,…,M, l = 1,2,…,L \}$ are $\kappa-\mu$ shadowed random variables. A $\kappa-\mu$ shadowed random variable $X$ with parameters $(\kappa,\mu,m,\Bar{x})$ has the following pdf \cite{paris2014statistical}:
	\begin{equation}
	f_X(x) = \frac{x^{\mu-1}}{\theta^{\mu-m}\lambda^m\Gamma[\mu]}e^{-\frac{x}{\theta}} \ _1F_1\left( m, \mu,\frac{x}{\theta} - \frac{x}{\lambda}\right), \ x \geq 0
	\label{kmu_pdf1} 
	\end{equation} where $\mathlarger{_1F_1(.)}$ is the confluent hypergeometric function, $\Gamma[.]$ is the gamma function, $\mathlarger{\theta = \frac{\bar{x}}{\mu(1+\kappa)}}$, $\mathlarger{\lambda = \frac{(\mu \kappa + m )\bar{x}}{\mu(1+\kappa)m}}$ and $\mathlarger{\bar{x} =}$ $ \mathlarger{\mathbb{E}[x]}$. Here, $\mathlarger{\mathbb{E}[.]}$ represents the expectation of a RV. Throughout this paper we assume that the CSIs of the links are not estimated frequently, but the statistics of the signal and interference links are known at the transmitters. 
	
	\section{Secondary user power control policy}
	In the underlay mode, the SU-Tx transmits over the same frequency used by the PU-Tx, even when the PU-Tx is active. Simultaneous transmission occurs as long as the quality of service (QoS) degradation at the PU-Rx due to interference from the SU-Tx is tolerable. This QoS degradation in the primary network is quantified by means of outage constraints at the PU-Rxs. Therefore, the SU-Tx must transmit at a power that keeps the outage at each of the PU-Rx below a predetermined level. Thus, transmit power policy at the SU-Tx can be mathematically formulated as follows \cite{louis_ergodic, badarneh_asymptotic},
	\begin{subequations}
		\begin{align}
		& \text{max } P_s,\\
		\text{s.t. } & \text{Pr}\{\gamma_{m,p}(P_s) \leq \gamma_0\} \leq p_0, \quad m=1,..., M \label{outage1}\\
		& P_s \leq P_{s,max}, 
		\end{align}
	\end{subequations}
	where $p_0$ is the maximum tolerable outage at each of the PU-Rx and $\gamma_0$ is the minimum desired SIR at the PU-Rx for a fixed PU transmit power $P_p$. The outage constraint in (\ref{outage1}) is equivalent to the condition where PU-Rx$_m$ with the lowest SIR satisfy the outage constraint. Hence, the power policy of SU-Tx can be alternatively formulated as
	\begin{subequations}
		\begin{align}
		&\text{max } P_s\\
		\text{s.t. } &  \text{Pr}\{\underset{1 \leq m \leq M}{\text{min}} \ \gamma_{m,p} (P_s) \leq \gamma_0\} \leq p_0  \ \label{tx_policy1}\\
		& P_s \leq P_{s,max}. \label{tx_policy2}
		\end{align}
	\end{subequations}
	Substituting the fading coefficients from (\ref{sir_p}) into (\ref{tx_policy1}), we obtain
	\begin{subequations}
		\begin{align}
		&\text{max } P_s\\
		\text{s.t. } &  \text{Pr} \left \lbrace \underset{1 \leq m \leq M}{\text{min}}  \frac{h_m}{\alpha_m} \leq \gamma_0 \frac{P_s}{P_p}  \right \rbrace  \leq p_0.
		\label{min_prob}\\
		&  P_s \leq P_{s,max}. \label{min_prob2}
		\end{align}
	\end{subequations}
	
	Here, $\{h_m;m=1,\cdots,M\}$ and $\{\alpha_m;m=1,\cdots,M\}$ are sequences of i.i.d. $\kappa-\mu$ shadowed RVs with fading parameters $(\kappa_{p},\mu_p,m_p, \bar h_p)$ and $(\kappa_{p,s},\mu_{p,s},$ $m_{p,s}, \bar \alpha_{p,s})$ respectively.Note that a more realistic model would rely on non-identical links between the transmitter and multiple receivers. However, analyzing this scenario is intractable due to the complex nature of the CCDF in generalized fading scenarios. The assumption of identical links holds true in scenarios where the users are in a stationary environment, such as ad-hoc networks in buildings or in case of slowly moving users \cite{louis_ergodic}. Similar, simplified models are widely used for the performance analysis of cognitive radio (CR) systems \cite{louis_ergodic,badarneh_asymptotic,ban2009multi,chen2006unified,song2006asymptotic}. The above-mentioned contributions analyze the performance of different CR systems assuming identical links between the transmitters and receivers. Therefore, even the study of the statistics of the minimum SIR over i.i.d. links is relevant and will hopefully serve as a spring-board for more general analysis.
	
	\par  To determine the optimum value of $P_s$ that satisfies the outage constraint in (\ref{tx_policy1}), we have to determine the CDF of the minimum of SIR RVs in a $\kappa-\mu$ shadowed fading environment. Note that, we can evaluate this using the CDF of the minimum of ratio of two $\kappa-\mu$ shadowed RVs as given in (\ref{min_prob}). The exact distribution of the minimum of any set of i.i.d. RVs $\gamma_{min} = min\{\gamma_1,\gamma_2,\cdots,\gamma_M \}$, where $\gamma_i \sim F_\gamma(z); \forall \  i \ \in \{1,\cdots,M \}$ is given by 
	\begin{equation}\label{actualcdf}
	F_{\gamma_{min}}(z) = 1-\left(1-F_\gamma(z) \right)^M.
	\end{equation}
	Hence, to evaluate the CDF in (\ref{min_prob}), we have to evaluate the $M^{th}$ power of the CCDF of ratio of $\kappa-\mu$ shadowed random variables. The exact expression for the CDF of ratio of $\kappa-\mu$ shadowed random variables is  given in  terms  of an infinite sum  of  the Lauricella’s  function  of  the  fourth  kind in \cite[Eq. 3]{kumar2019errata},\cite{kumar2017outage}. The complex nature of the CDF $F_\gamma(z)$ makes the evaluation of the $M^{th}$ power of the CDF difficult. Now, even if we found an approximation for the CDF of $\gamma$, any small error in the computation of $F_{\gamma}(z)$ will become amplified due to the exponent to which it is raised and hence it will make the corresponding distribution function less accurate. Note that even if we compute the exact distribution for large values of $M$, it will not be possible to derive any meaningful inference from them owing to the complex nature of those expressions. \\
	On the other hand, if we have a simple limiting distribution for (\ref{actualcdf}), which closely approximates the CDF values for moderate and large values of $M$, we can obtain a closed-form expression for the optimum $P_s$ that satisfies (\ref{min_prob}). For small values of $M$ we can still use the exact CDF of the minimum. Therefore, using tools from EVT, we formulate the following theorem to determine the limiting distribution of (\ref{actualcdf}), when $\gamma$ is the SIR in an $\kappa-\mu$ shadowed fading environment. We then use this theorem to evaluate the probability expression in (\ref{min_prob}) and hence obtain a closed-form expression for the optimum $P_s$. A similar approach is used for determining the ergodic multicast rate of the secondary users in \cite{badarneh_asymptotic} for Rayleigh faded channels. To the best of our knowledge, no previous work has used EVT to simplify the outage constraints at the PU-Rx. 
	
	\begin{theorem}\label{evt_main}
		Consider $K$ i.i.d. SIR RVs of the form
		\begin{equation}
		\gamma_{k} = \frac{|b_{k}|^2}{\sum\limits_{j=1}^{N}|c_{j,k}|^2},
		\label{sir_rv}
		\end{equation}
		where $\{|b_{k}|^2;1 \leq k \leq K, \}$  are i.i.d. $\kappa-\mu$ shadowed RVs with parameters $(\kappa, \mu, m, \bar x)$ and  $\{|c_{j,k}|^2; 1 \leq j \leq N \}$ are i.n.i.d. $\kappa-\mu$ shadowed random variables, with parameters $(\kappa_j, \mu_j, m_j, \bar x_j)$ $\forall \ k$, for $j=1,..,N$.  
		The asymptotic distribution of $\gamma^K_{min}= min(\gamma_1, \gamma_2,..., \gamma_K)$ is a Weibull distribution having the shape parameter $\upsilon=\mu$ and scale parameter $a_K = F^{-1}_\gamma\left( \frac{1}{K}\right)$, where $F_{\gamma}(z)$ is the common CDF of i.i.d. RVs $\gamma_k$. Let, $\gamma_{min} = \lim\limits_{K \to \infty}\gamma^K_{min}$, then we have, 
		
		\begin{equation}
		{F}_{ {\gamma}_{min}}(z) =  \begin{cases}
		1-exp(-(z/a_K)^{\upsilon}), & z \geq 0, \\
		0, & z<0.
		\end{cases} 
		\label{asymp_cdf}
		\end{equation}
	\end{theorem}
	\begin{proof}
		Please refer to Appendix \ref{minima_lim_cdf} for the proof. 
	\end{proof}
	
	To evaluate $a_K$, an approximation of the CDF $F_\gamma(z)$ relying on the Lauricella function of the forth kind given by \cite[Eqn 8]{kumar2017outage} is used. Furthermore, \cite{kumar2017outage} gives bounds on the truncation error and shows that the CDF is well approximated by the proposed expression. Finally, the MATLAB code for evaluating Lauricella’s function of the fourth kind is available in \cite{lauricode}.\\
	Note that the above expression is simpler to evaluate than the actual CDF of the minimum as given in (\ref{actualcdf}). Fig.\ref{kus_inid1} shows the simulated and theoretical asymptotic CDF of minimum over $K=20$ SIR RVs for different system parameters. Here, cases 1, 2 and 3 correspond to the channel fading parameters as given in Table \ref{sim_cdf_param}. The results indicate that the asymptotic results are close to the true minimum distribution even for the cases where the minimum is evaluated over moderate-length sequences, such as $K=20$. 
	
	\begin{table}[H]
		
		\centering
		\begin{tabular}{|c|c|c|c|c|c|c|c|}
			\hline
			Case \# & $\kappa$ & $\mu$ & $m$ & $N$ & $\{\kappa_i\}$ & $\{\mu_i\}$ & $\{m_i\}$ \\ \hline
			1 & 2 & 3 & 1 & 3 & $\{ 2,2,2\}$ & $\{ 2,2,2\}$ & $\{ 1,1,1\}$ \\ \hline
			2 & 2 & 3 & 1 & 2 & $\{ 2,2\}$ & $\{ 2,1\}$ & $\{ 1,1\}$ \\ \hline
			3 & 2 & 2 & 1 & 1 & $\{ 2\}$ & $\{ 1\}$ & $\{ 1\}$ \\ \hline
		\end{tabular}
		\caption{Simulation parameters used for Fig.\ref{kus_inid1}.}
		\label{sim_cdf_param}
	\end{table}

	Further, to better quantify mathematically the decrease in gap between the theoretical and simulated values  as $K$ increases, we have derived the rate of convergence of the asymptotic minimum distribution to the corresponding Weibull distribution. We now give the rate of convergence for our case through the following theorem.
	\begin{theorem} \label{rate_cnvg}
		The rate of convergence of $F_{\gamma^{K}_{min}}(z)$ to the Weibull distribution is \\ $\mathcal{O}\left(K^{-\mu ^{-1}} +  K^{-1} \right)$ where $\gamma^{K}_{min}=\min \{\gamma_1,\cdots,\gamma_K\}$ .
	\end{theorem} 
	\begin{proof}
		Please refer to Appendix \ref{rate_covg} for the proof.
	\end{proof}
	From this result, we observe that the rate of convergence depends on the length of the sequence $K$ and the source fading parameter $\mu$. The simulated and theoretical distribution are expected to be closer for large values of $K$. Further, the convergence will be faster for smaller values of $\mu$, the number of multi paths in the source to desired receiver link. 
	
	\par Using this asymptotic distribution, we can now determine the optimum $P_s$, when the number of PU-Rxs $M$, is moderate to large. To evaluate the CDF of $\gamma_{min,p} := \lim\limits_{M \to \infty} \underset{m}{\text{min}} \left \lbrace  \frac{h_m}{\alpha_m};m=1,\cdots,M \right \rbrace$ (to approximate (\ref{min_prob})), we now substitute $N=1$, $K=M$, $(\kappa, \mu, m, \bar x) = (\kappa_p, \mu_p, m_p, \bar h_p)$, $(\kappa_1, \mu_1, m_1, \bar x_1) = (\kappa_{p,s}, \mu_{p,s}, m_{p,s}, \bar \alpha_{p,s})$, $a_K=a_M=F_{{\gamma}}^{-1}\left(\frac{1}{M} \right)$ and  $\upsilon=\mu_{p}$ in Theorem \ref{evt_main}\footnote{Here, $F_\gamma(z)$ is evaluated using (\ref{cdf2}) for $N=1$. Even if we consider multiple primary interferers, note that Theorem 1 gives the asymptotic distribution of the minimum SIR for a case where the receiver suffers from the interference of $N$ other transmitters. Therefore, the theoretical framework developed is applicable for a much broader framework.  However, when we consider $(N-1)$ primary interferers having known transmit powers, the expression of the outage probability will be different and we will not have a closed form expression for the secondary user's power allocation. Furthermore, in cells having large cell radius, the interference arising from other primary transmitters can be neglected due to the associated high path loss.}. Hence, we have  
	
	\vspace{1mm}
	\begin{equation}
	1- \exp \left(-\left( \frac{\gamma_0P_s}{P_pa_M}\right)^{\mu_{p}} \right) \leq p_0.
	\label{outage_codtn}
	\end{equation} 
	
	Further rearrangement of (\ref{outage_codtn}) gives, 
	
	\begin{equation}
	P_s \leq \frac{P_p a_M}{\gamma_0}\left[-ln (1-p_0) \right]^{1/\mu_{p}}. 
	\end{equation}
	The largest $P_s$ that satisfies the above constraint is given by
	\begin{equation}
	P_s^+ = \frac{P_p a_M}{\gamma_0}\left[-ln (1-p_0) \right]^{1/\mu_{p}}.
	\label{power_sec}
	\end{equation}
	Now, using (\ref{power_sec}) and (\ref{tx_policy2}), the optimal $P_s$ for the SU-Tx power policy is given by
	\begin{equation}
	\Bar{P_s} = min\{P_s^+,P_{s,max} \}.
	\label{su_policy_final}
	\end{equation} 
	\begin{figure}[H]
		\centering
		\includegraphics[scale=0.5]{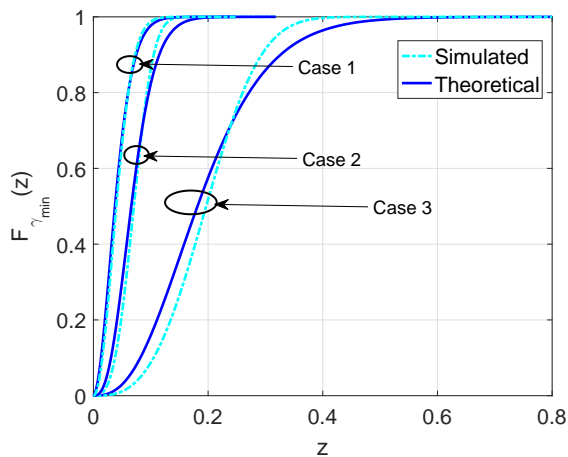}
		\caption{CDF of $\gamma^K_{min}$ (simulated) and $\gamma_{min}$ (using (\ref{asymp_cdf})) for different fading scenarios.}
		\label{kus_inid1}
	\end{figure}

	Now that we have derived the optimal SU-Tx power, we will analyze the impact of fading parameters on this power policy.
	From (\ref{su_policy_final}), we can observe that the optimum power at the SU-Tx, $\bar P_s$, is dependent on $P_s^+$ given in (\ref{power_sec}). The variations in $P_s^+$ are in turn governed by $P_P$, $\gamma_0$, $p_0$, $\mu_{p}$ and $a_M$.\\ 
	\textbf{\textit{Observation 1}}:
	From (\ref{power_sec}), it is plausible that an increase in either $P_p$ or $p_0$ or alternatively a decrease in $\gamma_0$ leads to an increase in $P_s^+$.\\
	
	The variation in $P_s^+$ with respect to the variations in the fading environment of the source can be studied by directly analyzing the variations in $a_M$ and $\mu_p$.
	However, the relationship between various fading parameters and $a_M$ is highly non-linear, therefore interpreting these variations with respect to changes in those parameters is difficult. One way to circumvent this problem is to use moment matching as in \cite{srinivasan2018secrecy}, and approximate each of the $\kappa-\mu$ shadowed RV as a gamma RV. \cite{subhash2019asymptotic} also uses similar approximation for analysis of their asymptotic maxima distribution. The $\kappa-\mu$ shadowed RV corresponding to the PU's fading coefficients with parameters ${(\kappa_p,\mu_p,m_p,\bar{h}_p)}$ can be approximated by a gamma RV having the shape parameter $\psi_1 = \frac{m_p\mu_p(1+\kappa_p)^2}{m_p+\mu_p\kappa_p^2+2m_p\kappa_p}$ and scale parameter $\psi_2=\frac{\bar{h}_p}{\psi_1}$. Similarly, the  $\kappa-\mu$ shadowed interferer (the interference from the SU-Tx) can also be  approximated as a gamma RV having shape parameter $\phi_1 = \frac{m_{p,s}\mu_{p,s}(1+\kappa_{p,s})^2}{m_{p,s}+\mu_{p,s}\kappa_{p,s}^2+2m_{p,s}\kappa_{p,s}}$ and scale parameter $\phi_2=\frac{\bar{\alpha}_{p,s}}{\phi_1}$. Hence, we have $\mathlarger{F_{\gamma}(z) = \mathbb{P}(\gamma \leq z)\approx \mathbb{P}\left( \frac{\boldsymbol{\Gamma}(\psi_1,\psi_2)}{\boldsymbol{\Gamma}(\phi_1,\phi_2)} \leq z \right)=\mathbb{P}\left(\frac{\boldsymbol{\Gamma}(\psi_1,1)}{\boldsymbol{\Gamma}(\phi_1,1)} \leq z \frac{\phi_2}{\psi_2} \right)}$, where $\mathlarger{\boldsymbol{\Gamma}(.,.)}$ represents a gamma distributed RV. This ratio of gamma RVs has a beta-prime CDF \cite{dubey1970compound} with parameters $\mathlarger{\psi_1}$ and $\mathlarger{\phi_1}$ evaluated at $\mathlarger{z\frac{\phi_2}{\psi_2}}$. Now, we use the theory of stochastic ordering to make inferences about the variations in $F_\gamma(z)$ with respect to variations in the channel fading conditions. The theory of stochastic ordering has been widely used by the statistical community to quantify the concept of one RV being greater than another or vice-versa. Furthermore, the authors of \cite{tepedelenlioglu2011applications,dhillon2013downlink,madhusudhanan2012stochastic,srinivasan2018secrecy} used stochastic ordering for the analysis of various wireless communication systems.\\ Here, we use the analysis in \cite{srinivasan2018secrecy} to make inferences about the approximate variation in $\mathlarger{F_{\gamma}(z)}$, with respect to the changes in $\kappa_p,\mu_p$, $m_p$, $\kappa_{p,s},\mu_{p,s}$ and $m_{p,s}$. {Although we have used the beta-prime approximation of the ratio of $\kappa-\mu$ shadowed RVs to provide an approximate analysis, using the same approximation to derive the minima distribution will be counterproductive. This is because, the CDF of the beta prime RV itself involves a Bessel function and the exact evaluation of the minimum using (\ref{actualcdf}) is still difficult. Secondly, if we try to derive the asymptotic  distribution using this approximate CDF, it will be less accurate due to approximations in computing $a_M$.}  Based on the analysis, we give the following observations:\\
	\textbf{\textit{Observation 2 } : $P_s^+$ increases upon increasing $\boldsymbol{\mu_p}$ or $\boldsymbol{m_p}$ or decreasing $\boldsymbol{\mu_{p,s}}$ or $\boldsymbol{m_{p,s}}$ .}\\
	Observe that, an increase in $\mu_p$ or $m_p$ results in an increase in $\psi_1$. According to $I4$ in Section III of \cite{srinivasan2018secrecy}, with an increase in $\psi_1$ along with a proportionate increase in $\bar{h}_p$, we can observe a reduction in $F_{\gamma}(z)$. Since the CDF is monotonically increasing function, to obtain the same CDF value of $\frac{1}{M}$ even after an increase in $\mu_p$ or $m_p$, the CDF evaluation point, which in our case is $a_M$, has to increase. Hence we infer from (\ref{power_sec}) that $P_s^+$ increases. \textbf{A similar argument can be made for a decrease in $\mu_{p,s}$ or $m_{p,s}$}.\\
	\textbf{\textit{Observation 3 } : $P_s^+$ increases upon increasing $\boldsymbol{\kappa_p}$ if $\boldsymbol{m_p-\mu_p \geq 0 }$ and decreases otherwise. Alternatively, $P_s^+$ increases upon decreasing $\boldsymbol{\kappa_{p,s}}$ if $\boldsymbol{m_{p,s}-\mu_{p,s} \geq 0 }$ and decreases otherwise.} \\
	The derivative of $\psi_1$ with respect to $\kappa_p$ is given by $\mathlarger{\frac{2\kappa_p(1+\kappa_p)m_p\mu_p(m_p-\mu_p)}{\left(m_p+2\kappa_p m_p+\kappa_p^2\mu_p \right)^2}}$. This shows that $\psi_1$ increases with an increase in $\kappa_p$ if $m_p-\mu_p>0$ and decreases otherwise. This in turn implies that the scale parameter $\mathlarger{F_{\gamma}^{-1}(M^{-1})}$ increases with an increase in $\kappa_p$, if $m_p-\mu_p > 0$ and decreases otherwise. Hence, following the same reasoning given in \textit{Observation 2}, we can infer that an increase in $\kappa_p$ increases $a_M$, if $m_p-\mu_p>0$, owing to the increase in $\psi_1$. Hence $P_s^+$ increases. Similarly, an increase in $\kappa_p$ results in an reduction of $a_M$, if $m_p-\mu_p<0$. Hence, $P_s^+$ decreases. Similarly, we can prove the opposite for change in $\kappa_{p,s}$.
	\par Thus \textit{Observation 2} and \textit{Observation 3} offers inferences on the variation of the maximum SU power $P_s^+$ with respect to the changes in the source and interferer fading environment. Furthermore, Table I of \cite{pozas_shadowed} summarizes the relationship between the $\kappa-\mu$ shadowed fading model and many common fading models, like Rayleigh, Rician, Nakagami etc. Hence, using these results we can analyze the variations for any specific fading environment as well.

	\section{Ergodic muticast rate of secondary users} \label{erg_cap}
	Here, multiple SU‐Rxs receive the same information from the SU‐Tx through a single radio transmission. Such multicast transmissions are useful for group‐based services such as audio‐video conferensing, disaster recovery, and military operations \cite{louis_ergodic}. The ergodic multicast rate of the secondary network is defined as \cite{ji2012capacity,louis_ergodic}
	\begin{equation}
	C_{sec}= L \times \mathbb{E}[\text{log}_2(1+ \underset{1 \leq l \leq L}{\text{min}} \gamma_{l,s})].
	\label{sec_min_cap}
	\end{equation}
	Substituting the expression for $\gamma_{l,s}$ from (\ref{sir_s}), we obtain
	\begin{equation}
	C_{sec}=  L \times \mathbb{E}[\text{log}_2(1+ \underset{1 \leq l \leq L}{\text{min }}\frac{P_s g_l}{P_p \beta_l})].
	\label{sec_min_cap1}
	\end{equation}
	Given that the CDF of the ratio of $\kappa-\mu$ shadowed RVs itself is complicated, it is a challenge to derive any simple expression for (\ref{sec_min_cap1}). Therefore, we propose the following theorem to evaluate the asymptotic ergodic multicast rate of SUs.
	\begin{theorem}\label{evt_rate_main}
		Consider $K$ i.i.d. SIR RVs of the form
		\begin{equation}
		\gamma_{k} = \frac{|b_{k}|^2}{\sum\limits_{j=1}^{N}|c_{j,k}|^2},
		\end{equation}
		where $\{|b_{k}|^2;1 \leq k \leq K, \}$  are i.i.d. $\kappa-\mu$ shadowed RVs with parameters $(\kappa, \mu, m, \bar x)$ and  $\{|c_{j,k}|^2; 1\leq j\leq N \}$ are i.n.i.d. $\kappa-\mu$ shadowed random variables, with parameters $(\kappa_j, \mu_j, m_j, \bar x_j)$ $\forall \ k$, for $j=1,..,N$. If $\gamma^K_{min}= min(\gamma_1, \gamma_2,..., \gamma_K)$, then
		\begin{equation}\label{ratelim}
		\lim\limits_{K \to \infty} \mathbb{E}[\text{log}_2(1+ \gamma_{min}^K)] = \mathbb{E}[{R}_{min}], 
		\end{equation}
		where $R_{min}=log_2(1+\gamma_{min})$ and $\gamma_{min}$ is the asymptotic distribution of $\gamma_{min}^K$ as given in Theorem (\ref{evt_main}).
	\end{theorem}
	\begin{proof}
		Please refer to Appendix \ref{rate_convg} for the proof. 
	\end{proof}
	The expectation in (\ref{sec_min_cap}) can now be evaluated using the pdf of the Weibull RV, whose CDF is given in (\ref{asymp_cdf}),  after substituting $N=1$, $K=L$, $(\kappa, \mu, m, \bar x) := (\kappa_s,\mu_s,m_s, \bar g_s)$, $(\kappa_1, \mu_1, m_1, \bar x_1) := (\kappa_{s,p},\mu_{s,p},m_{s,p}, \bar \beta_{s,p})$, $a_K=a_L=F_{{\gamma}}^{-1}\left(\frac{1}{L} \right)$ and  $\upsilon=\mu_{s}$. The asymptotic minimum ergodic multicast rate of the SU's is therefore given by 
	\begin{equation}
	C_{sec} \approx  L \times \int\limits_{0}^{\infty} \log_2\left(1+\frac{P_s x }{P_p}\right) \frac{\upsilon}{a_L} \left(\frac{x}{a_L} \right)^{\upsilon-1} \exp\left(-\left(\frac{x}{a_L} \right)^\upsilon \right) dx.
	\label{sec_min_cap_eval}
	\end{equation}
	To analyze the above expression with respect to $a_L$, we propose the following lemma.
	\begin{lemma}\label{weibulllemma1}
		Consider two Weibull RVs $P$ and $Q$ with parameters $\mathlarger{\{a_{L1}, \upsilon\}}$ and $\mathlarger{\{a_{L2}, \upsilon\}}$ respectively. $P$ is stochastically larger than $Q$ if
		\begin{equation}
		\mathbb P(P < z) < \mathbb P(Q <z), \ \forall z >0. 
		\end{equation}
		In other words, $\mathlarger{P>_{st}Q}$ if
		\begin{equation}
		1- exp\left(-\left( \frac{z}{a_{L1}}\right)^{\upsilon}\right) < 1-exp\left(-\left( \frac{z}{a_{L2}}\right)^{\upsilon}\right).
		\end{equation}
		The above condition is achieved when $\mathlarger{a_{L1} \geq a_{L2}}$.
		Also, if $\mathlarger{P>_{st}Q}$, $\mathbb{E}[u(P)] \geq \mathbb{E}[u(Q)]$ for any non-decreasing function $u$ \cite{shaked}.  
	\end{lemma}
	Note that the logarithm function is non-decreasing. Therefore, from the above lemma, we can conclude that the ergodic rate increases with the increase of $a_L$. Hence, we can make the following observations,  by following arguments similar to those made in \textit{Observation 2} and \textit{Observation 3}.\\
	\textbf{\textit{Observation 4 } : $C_{sec}$ increases upon increasing $\boldsymbol{m_s}$ or decreasing $\boldsymbol{m_{s,p}}$ or $\boldsymbol{\mu_{s,p}}$.}\\
	\textbf{\textit{Observation 5 } : $C_{sec}$ increases upon increasing $\boldsymbol{\kappa_{s}}$ if $\boldsymbol{m_{s}-\mu_{s} \geq 0 }$ and decreases otherwise. Alternatively, $C_{sec}$ increases upon decreasing $\boldsymbol{\kappa_{s,p}}$ if $\boldsymbol{m_{s,p}-\mu_{s,p} \geq 0 }$ and decreases otherwise.} \\
	\textbf{\textit{Observation 6 } : Also, $C_{sec}$ is directly proportional to $P_s^+$. Hence, variation in $C_{sec}$ with respect to the variations in the fading channel of the primary network can be directly extended from \textit{Observation 2} and \textit{Observation 3}.}

	\section{Numerical results and Simulations}
	
	In this section we present simulations to validate the results and observations from the previous sections. The PU-Tx's target rate is chosen to be $R_{0}=\log(1+\gamma_0)=0.03$ bps/Hz for all the simulations. This is to match the performance target for the operational long-term evolution (LTE) network, which requires the cell edge user throughput to be higher than $0.02$ bps/Hz/cell/user\cite{louis_ergodic,itu2008requirements,sesia2011lte}. Similarly, all the results are generated for the choice of $P_{s,max}=20 \ dB$. Here, Fig. \ref{sec_alloc_po} and \ref{sec_alloc_pp} show the SU-Tx power allocation for various combinations of PU-Tx power $P_p$ and PU-Rx outage constraint $p_0$. Furthermore, we have chosen  $(\kappa_p=3,\mu_p=2,m_p=1)$, $(\kappa_{p,s}=2,\mu_{p,s}=2,m_{p,s}=1)$, $(\kappa_s=2,\mu_s=2,m_s=1)$, $(\kappa_{s,p}=3,\mu_{s,p}=3,m_{s,p}=1)$, $M=10$ and $L=10$ for generating Figs. \ref{sec_alloc_po}-\ref{sec_alloc_out_sec_pp}. The results indicate that the optimum SU-Tx power $\bar{P_s}$ increases upon increasing the PU-Tx power $P_p$. This is because, upon increasing $P_p$, the PU-Rxs become capable of handling a higher interference arriving from the SU-Tx at the same outage constraints. Furthermore, for constant $P_p$, $\bar{P_s}$ decreases with a reduction in $p_0$. This is because, a reduction in $p_0$ results in stricter outage constraints at the PU-Rxs.  In order to satisfy these stricter reliability conditions, the SU-Tx has to transmit at a lower power for reducing the interference at the PU-Rx. Note that the optimum transmit power $\bar{P_s}$ is always limited by $P_{s,max}$.  For the power allocation considered in Fig. \ref{sec_alloc_po}, we show the outage probabilities of both the primary and of the secondary receiver having lowest SIR in Fig. \ref{sec_alloc_out_pri} and \ref{sec_alloc_out_sec} respectively. Similarly, in Fig. \ref{sec_alloc_out_pri_pp} and \ref{sec_alloc_out_sec_pp} we show the outage probabilities of the primary and secondary receiver having the lowest SIR for the power allocation considered in Fig. \ref{sec_alloc_pp}. From, Fig. \ref{sec_alloc_out_pri}-\ref{sec_alloc_out_sec_pp} we observe that for a constant value of $p_0$, the secondary user power allocation ensures that the outage of the PU-Rx and SU-Rx having the lowest SIR remains constant with respect to $P_p$. However, note that we are not constraining the outage probability of the secondary users in the allocation scheme and hence the probability of outage of the secondary users may change with the channel conditions or system model.
	
	\begin{figure}[H]
		\centering
		\begin{minipage}[t]{0.45\textwidth}
			\includegraphics[scale=0.48]{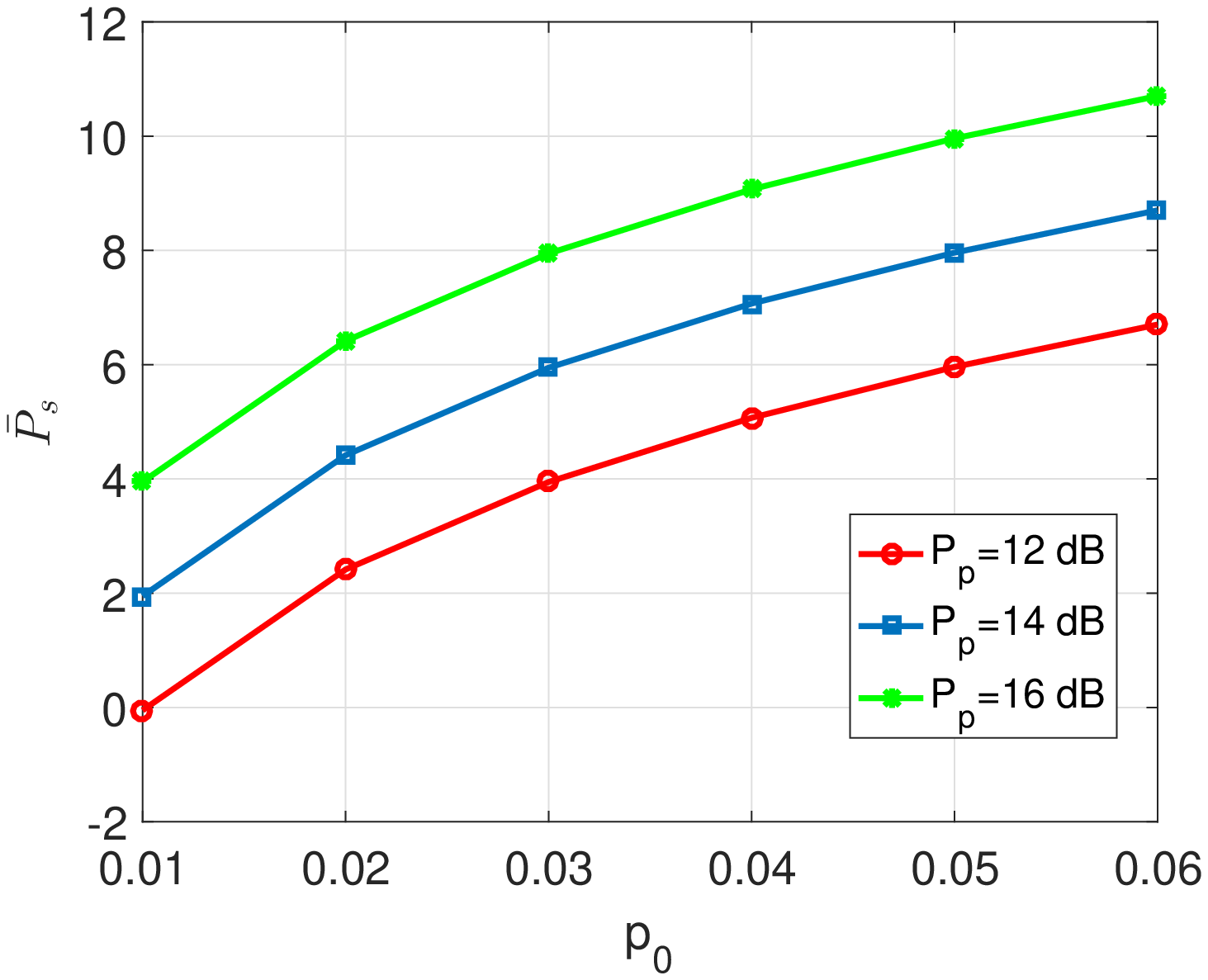}
			\caption{$p_0$ vs $\Bar{P}_s$ for M=10 and L=10. }
			\label{sec_alloc_po}	
		\end{minipage}
		\begin{minipage}[t]{0.45\textwidth}
			\includegraphics[scale=0.48]{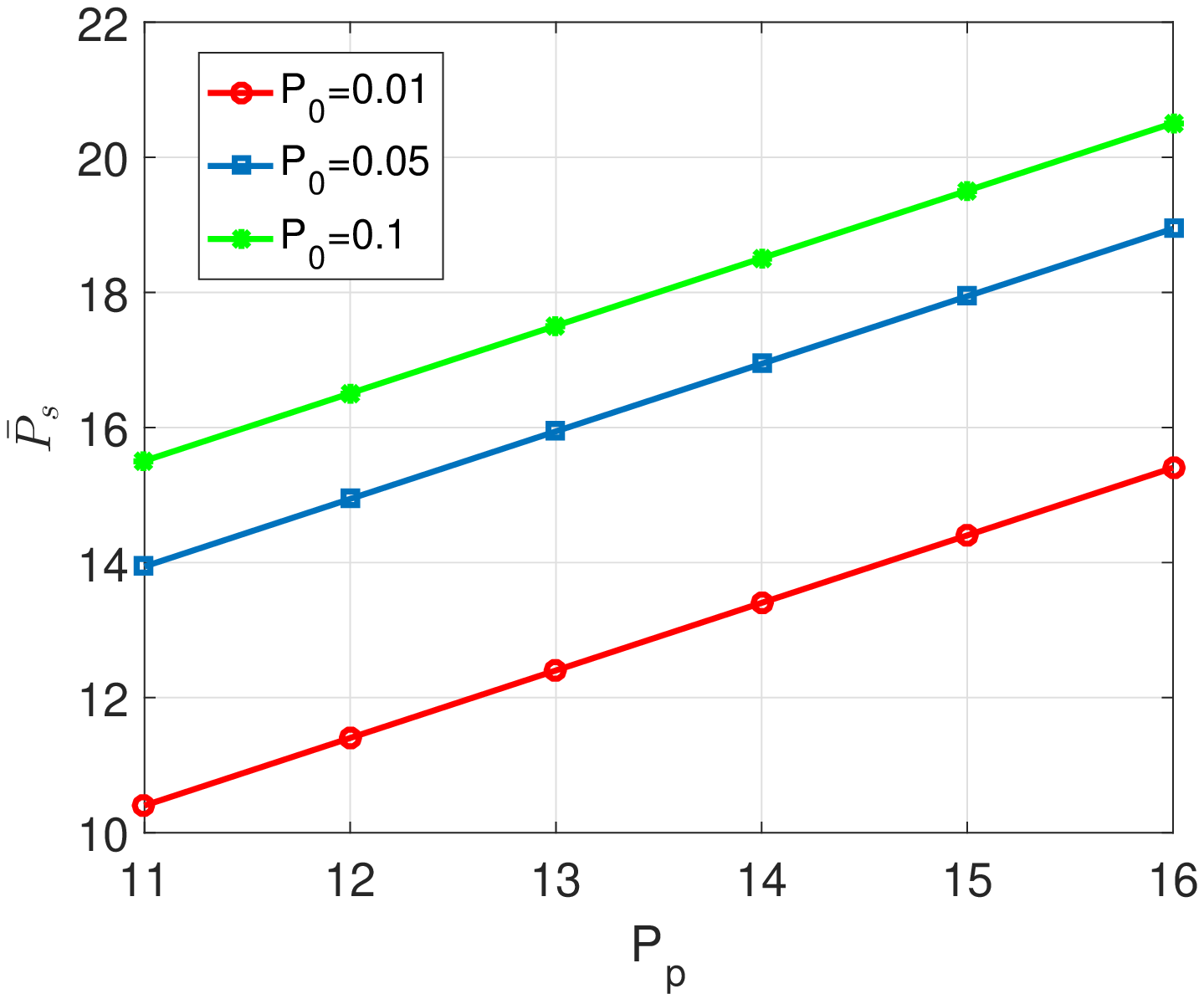}
			\caption{$P_p$ vs $\Bar{P}_s$ for M=10 and L=10. }
			\label{sec_alloc_pp}
		\end{minipage}%
	\end{figure}

	\begin{figure}[H]
		\centering
		\begin{minipage}[t]{0.45\textwidth}
			\includegraphics[scale=0.5]{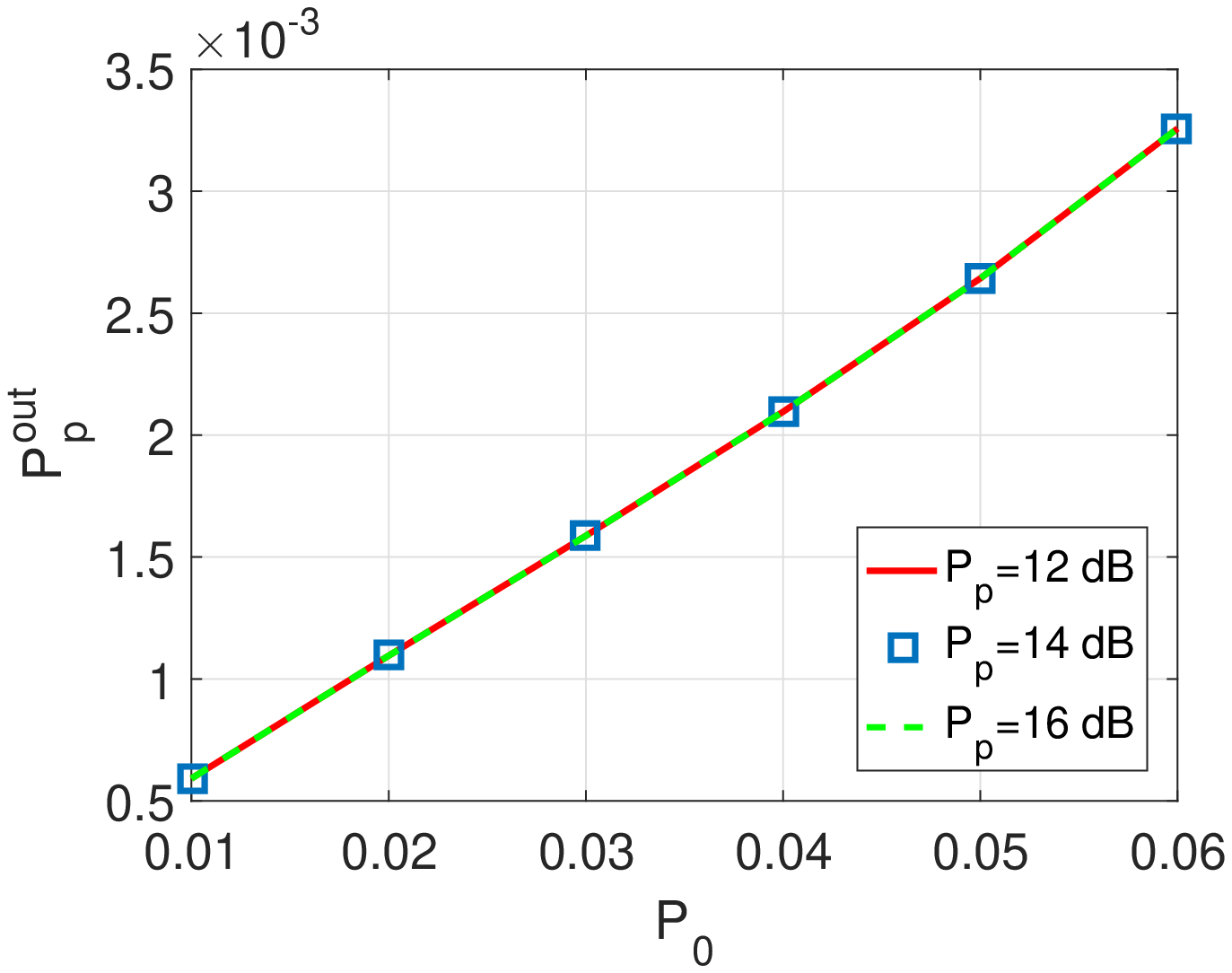}
			\caption{$p_o$ vs outage probability of minimum SIR primary user. }
			\label{sec_alloc_out_pri}
		\end{minipage}
		\begin{minipage}[t]{0.45\textwidth}
			\includegraphics[scale=0.5]{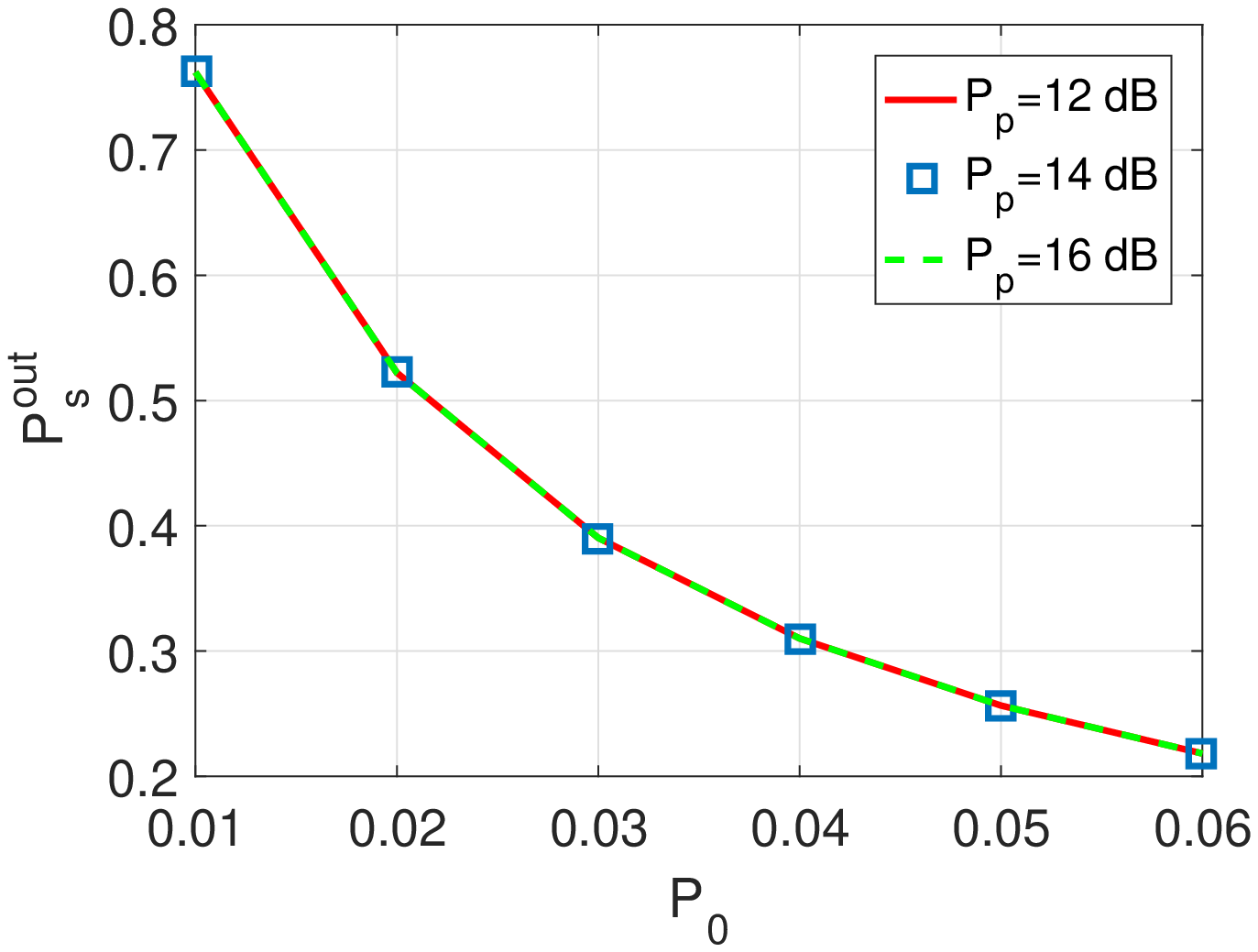}
			\caption{$P_p$ vs  outage probability of minimum SIR primary user. }
			\label{sec_alloc_out_sec}
		\end{minipage}%
	\end{figure}
	
	\begin{figure}[H]
		\centering
		\begin{minipage}[t]{0.45\textwidth}
			\includegraphics[scale=0.5]{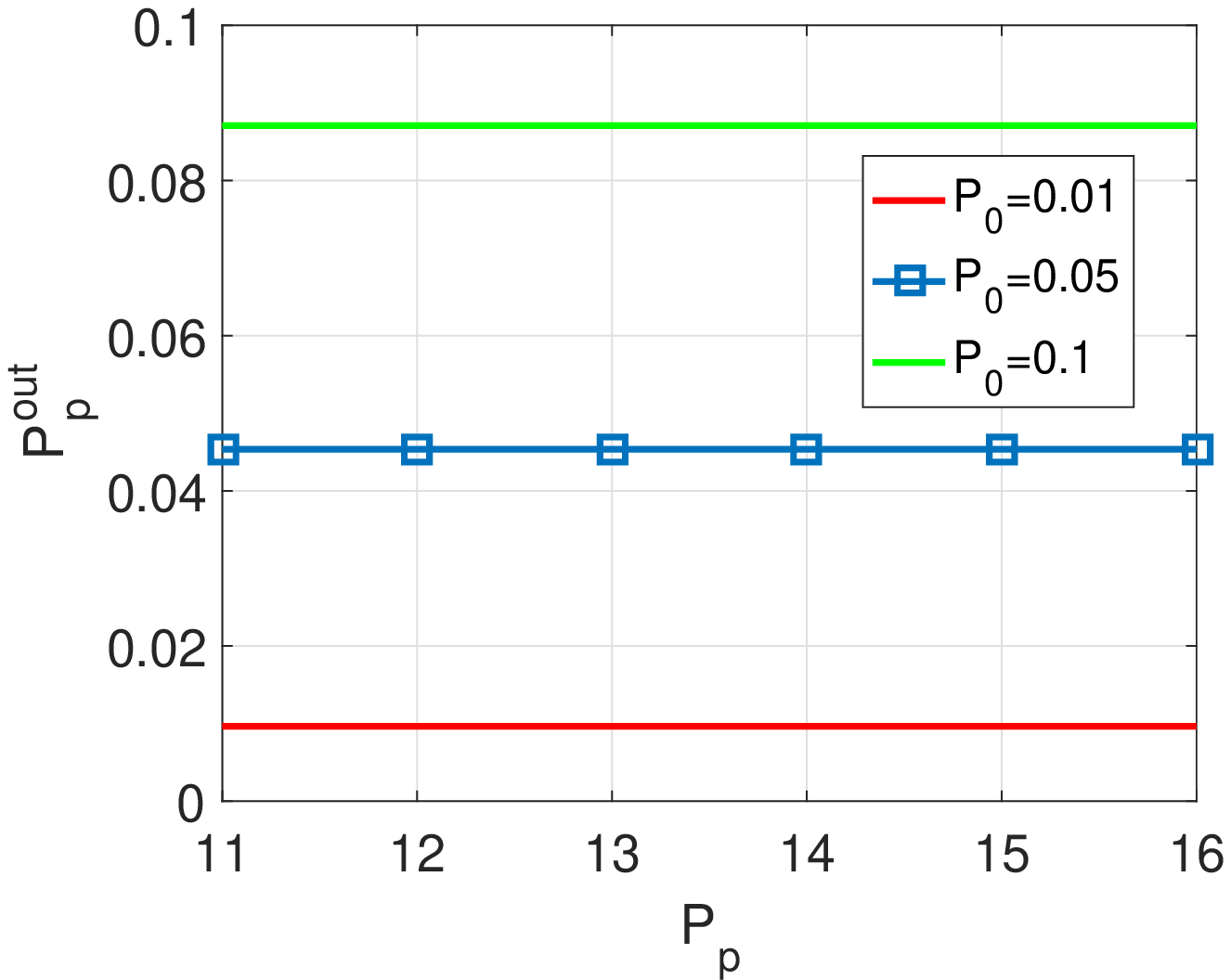}
			\caption{$p_0$ vs  outage probability of minimum SIR secondary user. }
			\label{sec_alloc_out_pri_pp}
		\end{minipage}
		\begin{minipage}[t]{0.45\textwidth}
			\includegraphics[scale=0.5]{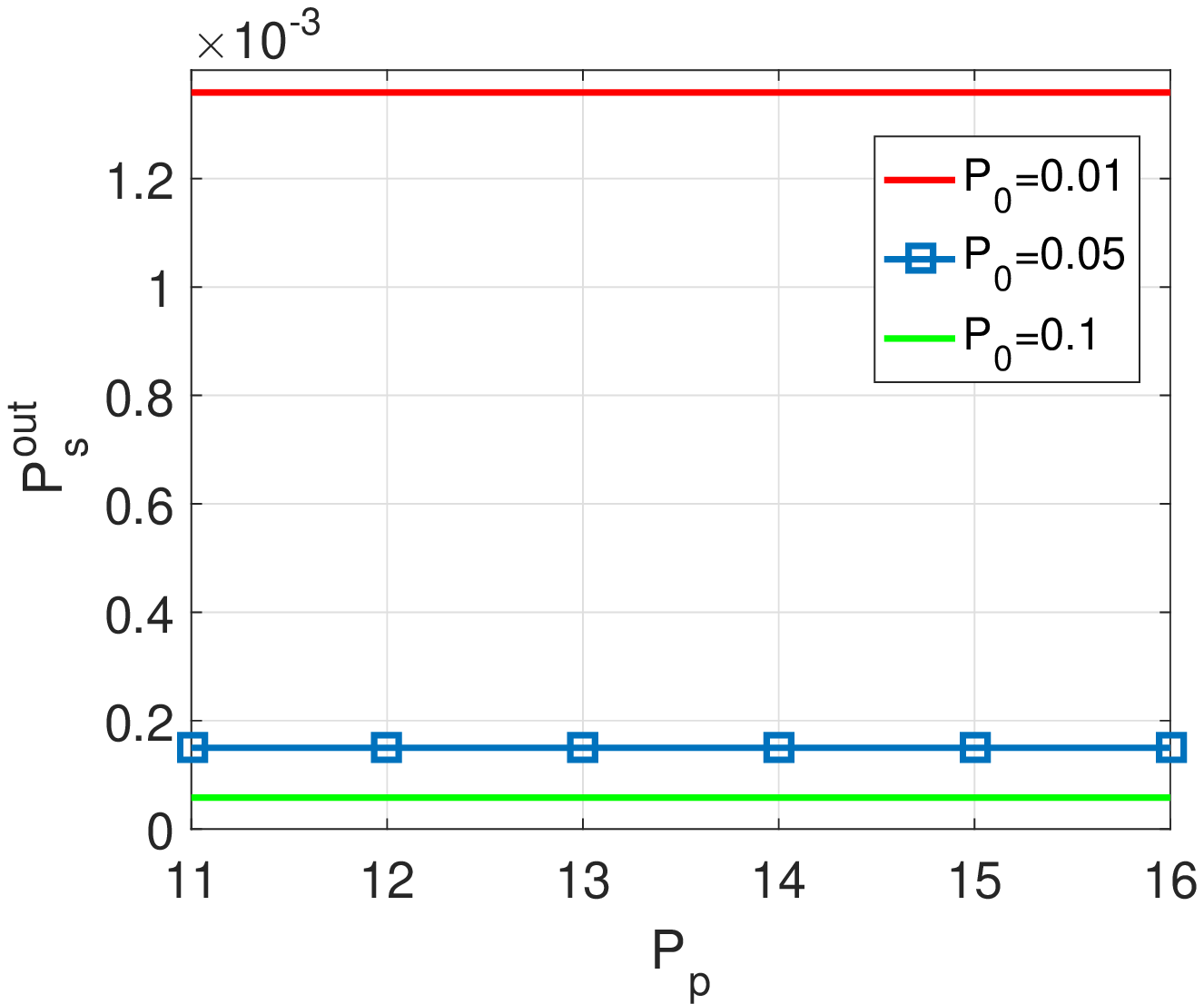}
			\caption{$P_p$ vs outage probability of minimum SIR secondary user. }
			\label{sec_alloc_out_sec_pp}
		\end{minipage}%
	\end{figure}

	\begin{figure}[H]
		\centering
		\begin{minipage}[t]{0.45\textwidth}
			\includegraphics[scale=0.5]{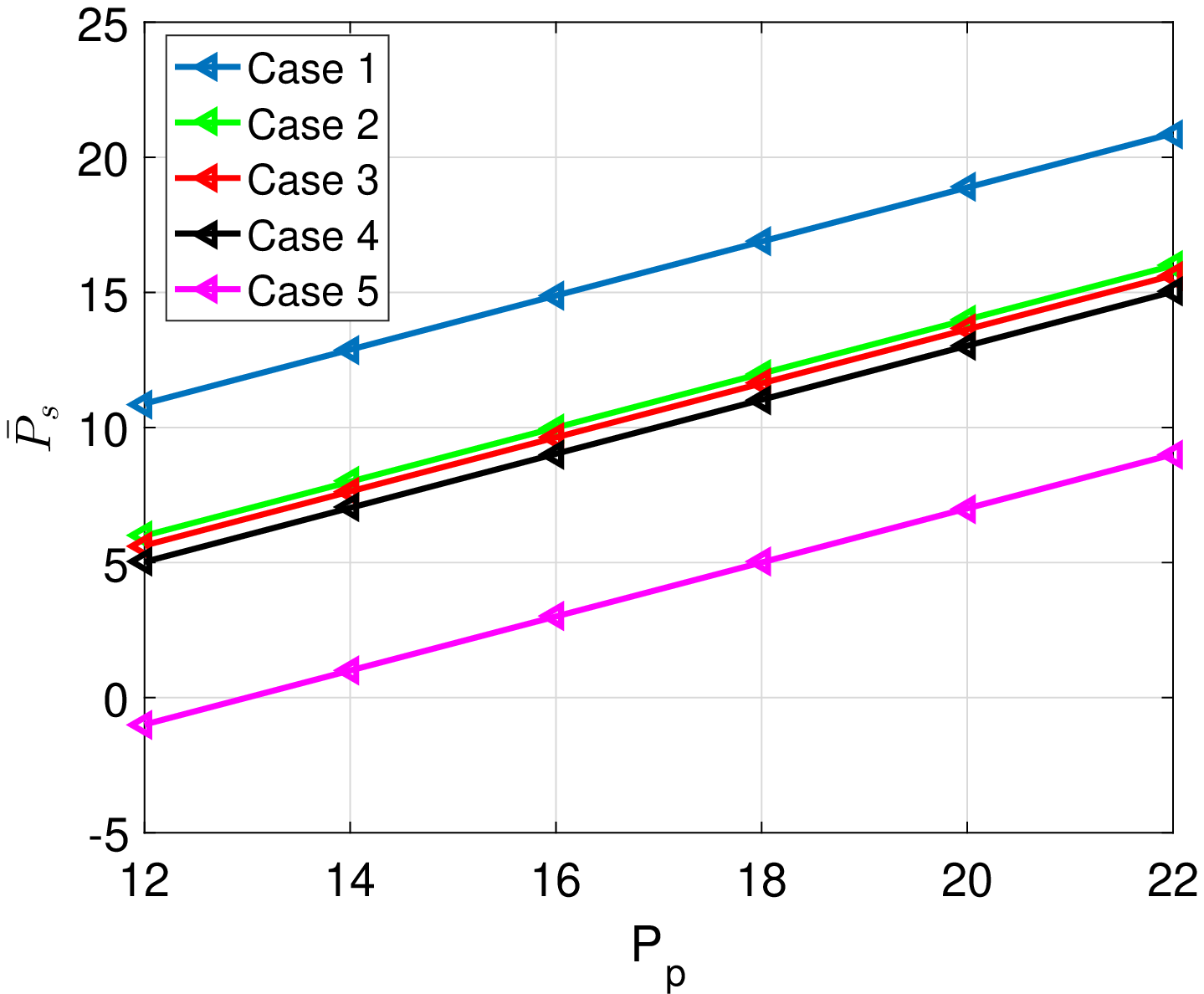}
			\caption{$P_p$ vs $\bar{P_s}$ for M=20.}
			\label{ps_vary}
		\end{minipage}
		\begin{minipage}[t]{0.45\textwidth}
			\includegraphics[scale=0.5]{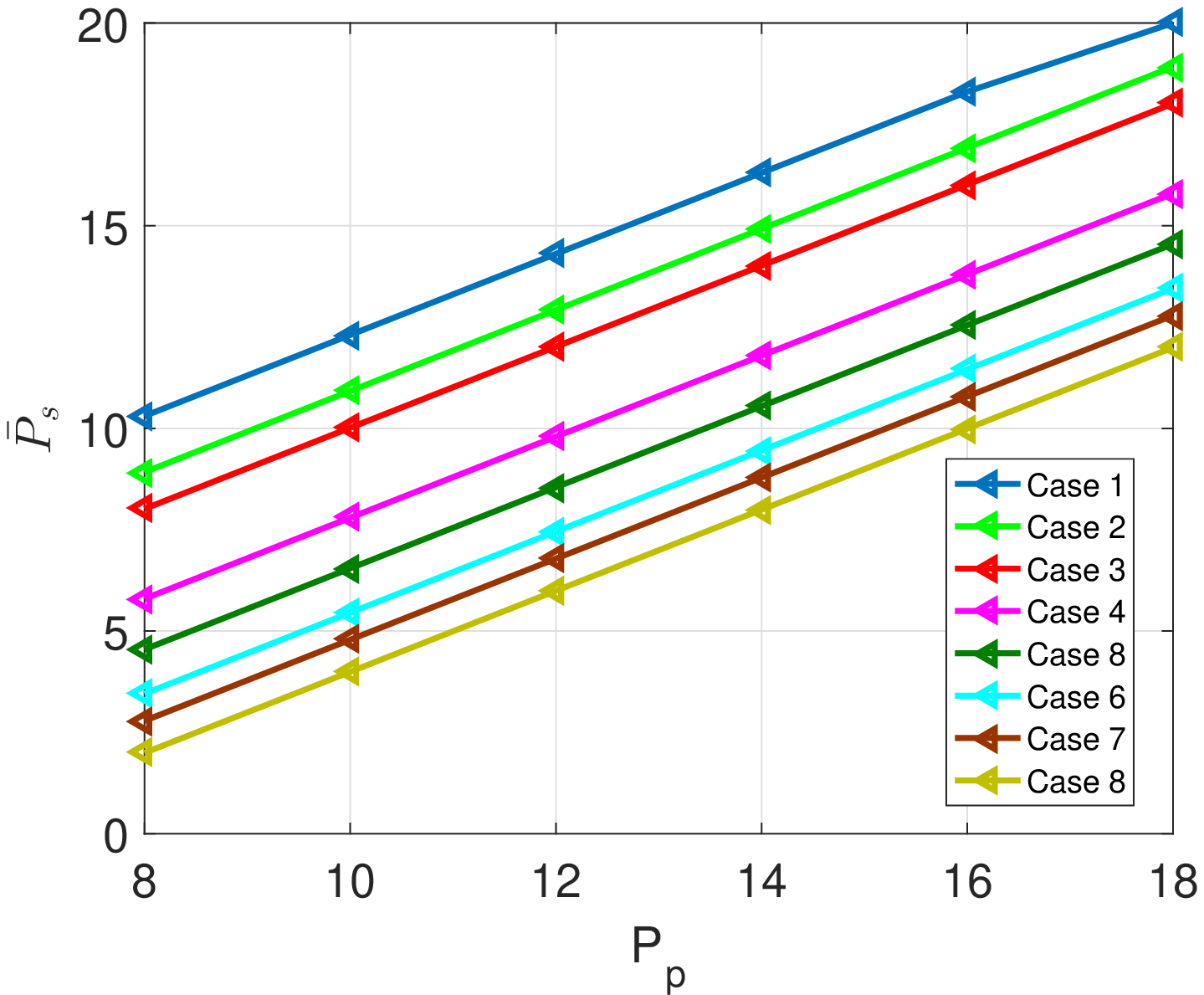}
			\caption{$P_p$ vs $\bar{P_s}$ for M=20.}
			\label{ps_vary_k_p_kps}
		\end{minipage}%
	\end{figure}
	
	Fig. \ref{ps_vary} shows plots of $P_p$ versus $\bar{P_s}$ for different channel conditions to validate \textit{Observation 2}. The channel parameters corresponding to the cases shown in the figure are given in Table \ref{observ_2_tb1}. From Cases 1,2 and 5 we can observe an increase in $\bar{P_s}$ with an increase in $\mu_p$ and $m_p$. Similarly, we can observe an decrease in $\bar{P_s}$ with an increase in $\mu_{p,s}$ and $m_{p,s}$ from cases 2,3 and 4. Next, \textit{Observation 3} is validated using simulations in Fig. \ref{ps_vary_k_p_kps}. The channel fading fading parameters used for simulation are given in Table \ref{observ_2_tb2}. According to \textit{Observation 3}, variation in $\bar{P_s}$ with changes in $\kappa_p$ or $\kappa_{p,s}$ depends upon the sign of $\mu_p-m_p$ and $\mu_{p,s}-m_{p,s}$ respectively. We verify all such variations possible using cases 1-8 in Fig. \ref{ps_vary_k_p_kps}. 
	The above figures validate the claim that the proposed asymptotic results can be used to derive inferences on the system performance easily. Without the proposed simple distribution for the minimum SIR RV, predicting the changes in the underlay CRN performance with respect to variations in channel fading conditions would have been non trivial.
	
	\begin{table}[H]
		\centering
		\begin{minipage}[b]{0.5\textwidth}
			\begin{tabular}{|c|c|c|c|c|c|c|}
				\hline
				Case \# & $\kappa_p$ & $\mu_p$ & $m_p$ & $\kappa_{p,s}$ & $\mu_{p,s}$ & $m_{p,s}$ \\ \hline
				1       & 3          & 2       & 1     & 2              & 1           & 1         \\ \hline
				2       & 3          & 1       & 1     & 2              & 1           & 1         \\ \hline
				3       & 3          & 1       & 1     & 2              & 20          & 1         \\ \hline
				4       & 3          & 1       & 1     & 2              & 1           & 10        \\ \hline
				5       & 3          & 1       & 0.1   & 2              & 1           & 1         \\ \hline
			\end{tabular}
			\caption{Channel parameters used for simulation of Fig.\ref{ps_vary}.}
			\label{observ_2_tb1}
		\end{minipage}
		\begin{minipage}[b]{0.5\textwidth}
			\begin{tabular}{|c|c|c|c|c|c|c|}
				\hline
				Case \# & $\kappa_p$ & $\mu_p$ & $m_p$ & $\kappa_{p,s}$ & $\mu_{p,s}$ & $m_{p,s}$ \\ \hline
				1       & 3          & 2       & 1     & 2              & 1           & 1         \\ \hline
				2       & 10         & 2       & 1     & 2              & 1           & 1         \\ \hline
				3       & 3          & 1       & 10    & 2              & 1           & 2         \\ \hline
				4       & 3          & 1       & 10    & 10             & 1           & 2         \\ \hline
				5       & 10         & 1       & 2     & 2              & 1           & 1         \\ \hline
				6       & 3          & 1       & 2     & 2              & 1           & 1         \\ \hline
				7       & 3          & 1       & 1     & 30             & 2           & 1         \\ \hline
				8       & 3          & 1       & 1     & 2              & 2           & 1         \\ \hline
			\end{tabular}
			\caption{Channel parameters used for simulation of Fig.\ref{ps_vary_k_p_kps}}
			\label{observ_2_tb2}
		\end{minipage}%
	\end{table}
	
	Next, in Figs. \ref{cap_1_pri}-\ref{cap_k_vary_sp} we compare the simulated and theoretical values of the ergodic multicast rate of secondary users under different channel conditions. We had chosen $P_p=14$ dB and $p_0=0.1$ for all these figures.  Fig. \ref{cap_1_pri} shows the variation in $C_{sec}/L$ with respect to variation in number of primary users $M$, for two different values of $p_0$. The channel fading parameters chosen for simulation are given in Table. \ref{observ_4_tb1}. \textit{Observations} 4-5 are verified via simulation in Figs. \ref{cap_mu_m_vary}-\ref{cap_k_vary_sp}.  Fig. \ref{cap_mu_m_vary} gives simulations to validate \textit{Observation} 4 and the corresponding channel parameters are given in Table \ref{obs_4}.  Comparing, Fig. \ref{cap_1_pri} and Fig. \ref{cap_mu_m_vary} we observe that the convergence of the simulated value of minimum rate to the asymptotic value proposed is slower for larger values of $\mu_s$. This agrees with our rate of convergence results discussed in Theorem  \ref{rate_cnvg}. Similarly, \textit{Observation} 5 is verified in Fig. \ref{cap_k_vary_3} and \ref{cap_k_vary_sp}. The corresponding channel parameters are given in Table \ref{observ_4_tb1}.  The difference between values of $C_{sec}/L$ for different values of $\kappa_s$ and $\kappa_{s,p}$ is not very large and hence for clarity of figures, we include only the simulated curves for minimum secondary capacity.

	\begin{figure}[H]
		\centering
		\begin{minipage}[t]{0.45\textwidth}
			\includegraphics[scale=0.5]{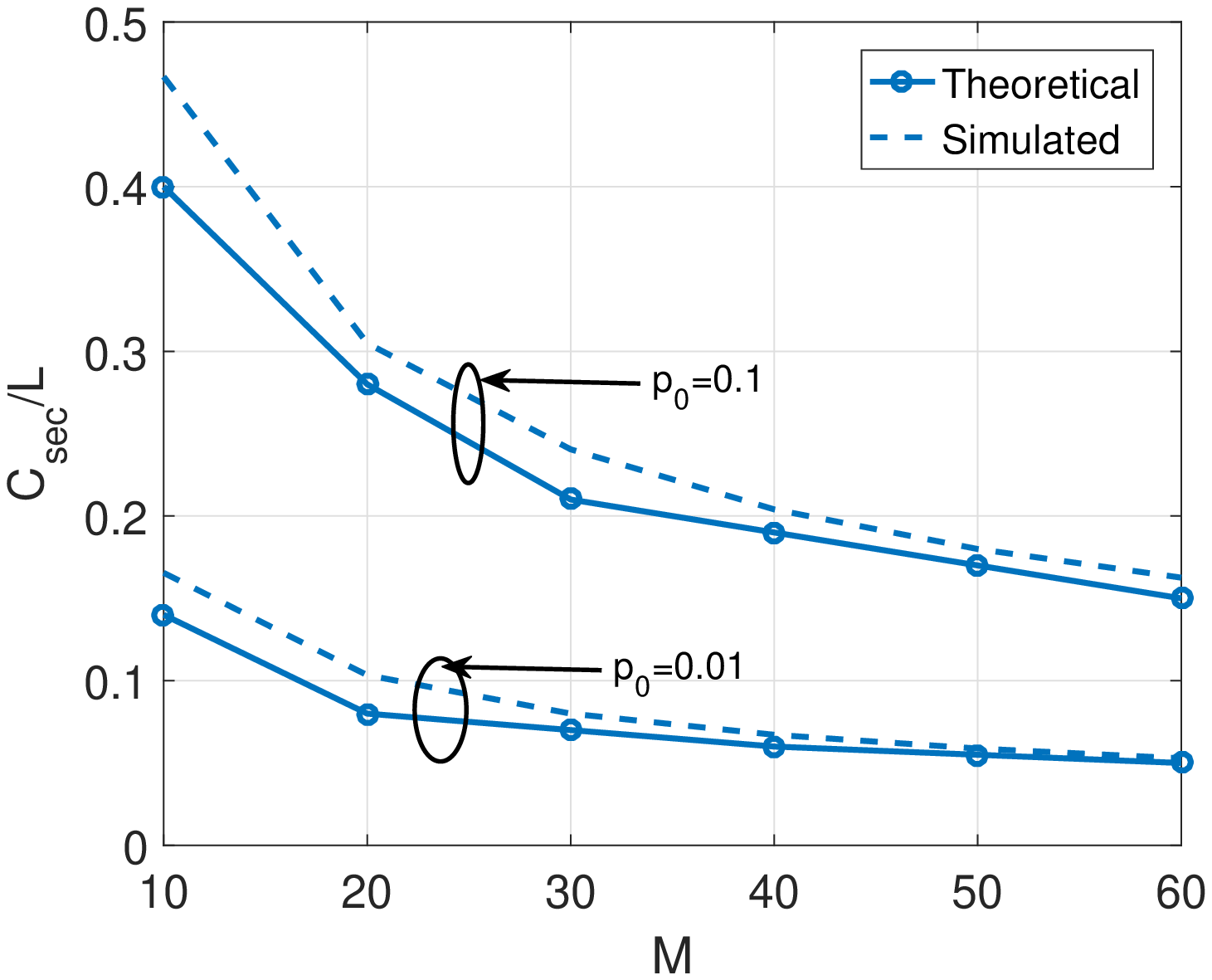}
			\caption{M vs $C_{sec}/L$. }
			\label{cap_1_pri}	
		\end{minipage}
		\begin{minipage}[t][][t]{0.45\textwidth}
			\includegraphics[scale=0.45]{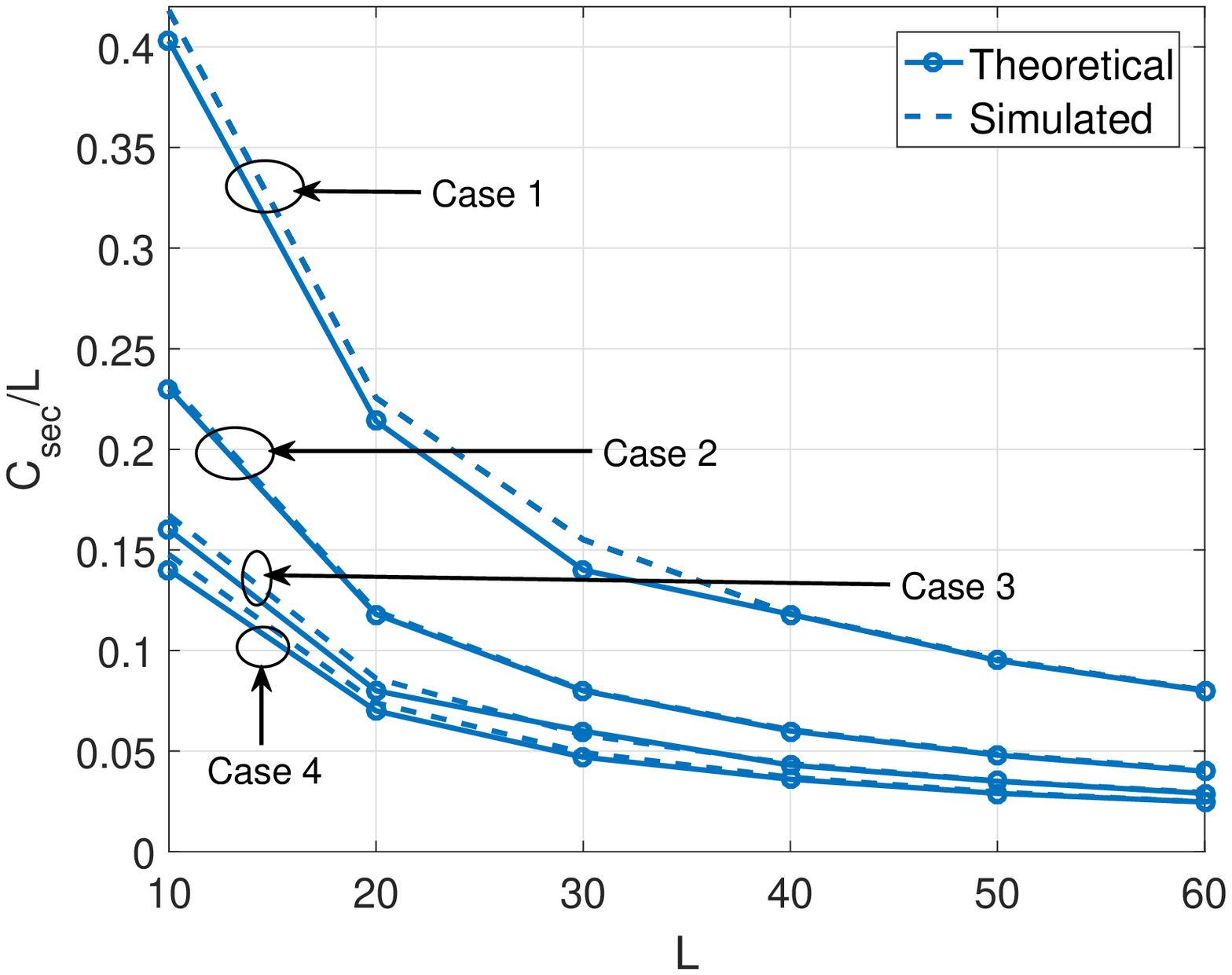}
			\caption{L vs $C_{sec}/L$. }
			\label{cap_mu_m_vary}	
		\end{minipage}
	\end{figure}

	\begin{table}[H]
		\centering
		\begin{minipage}[b]{0.55\textwidth}
			\begin{tabular}{|c|c|c|c|c|c|c|c|}
				\hline
				Fig \# & $\kappa_p$ & $\mu_p$ & $m_p$ & $\kappa_{p,s}$ & $\mu_{p,s}$ & $m_{p,s}$ & $M$ \\ \hline
				10 & 2 & 3 & 1 & 2 & 2 & 1 & - \\ \hline
				12 & 3 & 2 & 1 & 2 & 1 & 10 & 20 \\ \hline
				13 & 3 & 2 & 1 & 2 & 1 & 10 & 20 \\ \hline
				Fig \# & $\kappa_s$ & $\mu_s$ & $m_s$ & $\kappa_{s,p}$ & $\mu_{s,p}$ & $m_{s,p}$ & $L$ \\ \hline
				10 & 2 & 2 & 1 & 3 & 3 & 1 & 20 \\ \hline
				12 & - & - & - & 2 & 1 & 1 & - \\ \hline
				13 & 3 & 1 & 1 & - & - & - & - \\ \hline
			\end{tabular}
			\caption{Channel parameters used for simulation of Fig.\ref{cap_1_pri},\ref{cap_k_vary_3} and \ref{cap_k_vary_sp}}
			\label{observ_4_tb1}
		\end{minipage}
		\begin{minipage}[b]{0.45\textwidth}
			\begin{tabular}{|c|c|c|c|c|c|c|}
				\hline
				Case \# & $\kappa_s$ & $\mu_s$ & $m_s$ & $\kappa_{s,p}$ & $\mu_{s,p}$ & $m_{s,p}$ \\ \hline
				1       & 3          & 1       & 10    & 2              & 1           & 1         \\ \hline
				2       & 3          & 1       & 1     & 2              & 1           & 1         \\ \hline
				3       & 3          & 1       & 1     & 2              & 2           & 1         \\ \hline
				4       & 3          & 1       & 1     & 2              & 1           & 10        \\ \hline
			\end{tabular}
			\caption{Channel parameters used for Fig.\ref{cap_mu_m_vary}.}
			\label{obs_4}
		\end{minipage}%
	\end{table}
	
	\begin{figure}[H]
		\centering
		\begin{minipage}[t][][t]{0.45\textwidth}
			\includegraphics[scale=0.45]{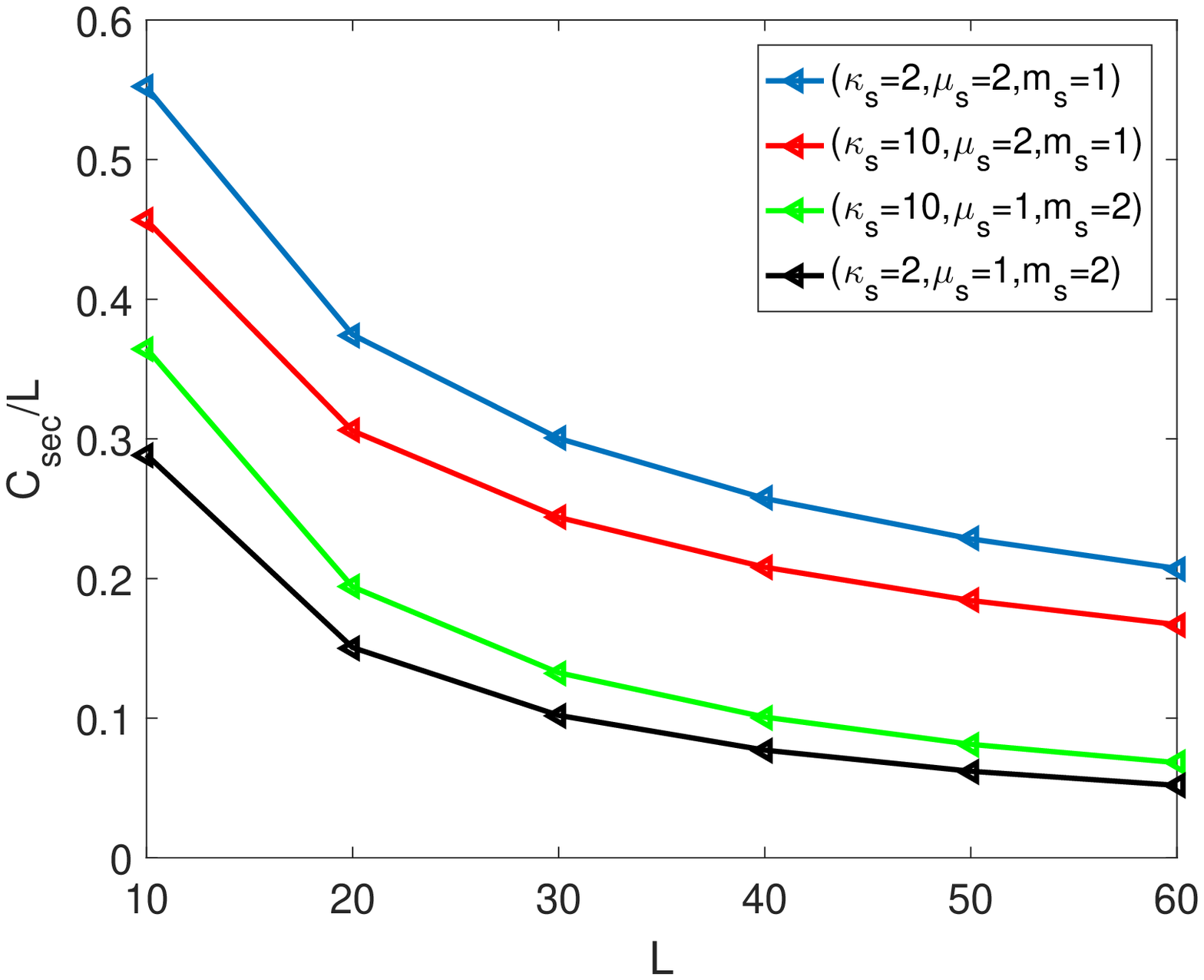}
			\caption{$L$ vs $C_{sec}/L$.}
			\label{cap_k_vary_3}	
		\end{minipage}
		\begin{minipage}[t][][t]{0.45\textwidth}
			\includegraphics[scale=0.45]{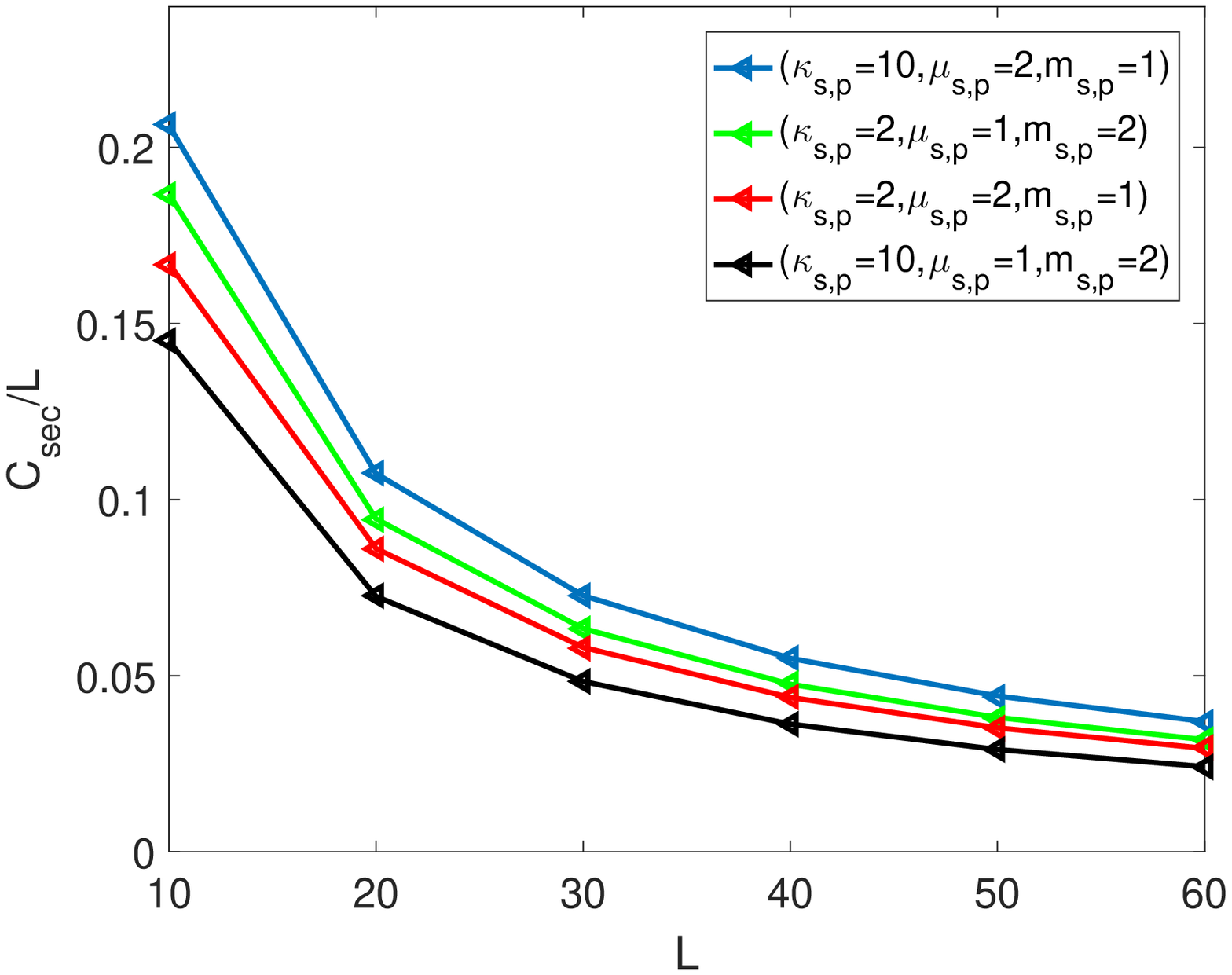}
			\caption{$L$ vs $C_{sec}/L$.}
			\label{cap_k_vary_sp}	
		\end{minipage}
	\end{figure} 
	
	\section{Summary}
	
	To summarize, we make use of tools from EVT to characterize the asymptotic distribution of the minimum of the ratio of $\kappa-\mu$ shadowed random variables and hence derive a simple expression for the distribution of the minimum SIR of PU/SU in a CR environment. We also derive the rate of convergence of the actual distribution of the minimum SIR to the derived distribution. These results are further used to find the optimal SU power allocation and the ergodic multicast rate of SUs. Assuming all the links are undergoing $\kappa-\mu$ shadowed fading, we use results from stochastic ordering to analyze the impact of various channel parameters on the SU performance and have the following analytical observations\footnote{Here, $\uparrow$ and $\downarrow$ are used to represent increase and decrease respectively.}:
	
	\begin{table}[H]
		\centering
		\begin{tabular}{|c|c|c|}
			\hline
			\textbf{\begin{tabular}[c]{@{}c@{}}Increase in \\ system parameter\end{tabular}} & \textbf{\begin{tabular}[c]{@{}c@{}}Optimum\\  SU-Tx power $(\bar{P_s})$\end{tabular}} & \textbf{\begin{tabular}[c]{@{}c@{}}Ergodic MC rate \\ of  SU ($C_{sec}/L$)\end{tabular}} \\ \hline
			$P_p$ & $\uparrow$ & $\uparrow$ \\ \hline
			$p_0$ & $\uparrow$ & $\uparrow$ \\ \hline
			$\gamma_0$ & $\downarrow$ & $\downarrow$ \\ \hline
			$M$ & $\downarrow$ & $\downarrow$ \\ \hline
			$L$ & $-$ & $\downarrow$ \\ 
			\hline
			\begin{tabular}[c]{@{}c@{}}$\kappa_p$, if $m_p-\mu_p>0$\end{tabular} & $\uparrow$ & $\uparrow$ \\ 
			\hline
			\begin{tabular}[c]{@{}c@{}}$\kappa_p$, if $m_p-\mu_p<0$\end{tabular} & $\downarrow$ & $\downarrow$ \\ \hline
			$\mu_p$ & $\uparrow$ & $\uparrow$ \\ \hline
			$m_p$ & $\uparrow$ & $\uparrow$ \\ \hline
			\begin{tabular}[c]{@{}c@{}}$\kappa_{p,s}$, if $m_{p,s}-\mu_{p,s}>0$\end{tabular} & $\downarrow$ & $\downarrow$ \\ \hline
			\begin{tabular}[c]{@{}c@{}}$\kappa_{p,s}$, if $m_{p,s}-\mu_{p,s}<0$\end{tabular} & $\uparrow$ & $\uparrow$ \\ \hline
			$\mu_{p,s}$ & $\downarrow$ & $\downarrow$ \\ \hline
			$m_{p,s}$ & $\downarrow$ & $\downarrow$ \\ \hline
			\begin{tabular}[c]{@{}c@{}}$\kappa_s$, if $m_s-\mu_s>0$\end{tabular} & $-$ & $\uparrow$ \\ 
			\hline
			\begin{tabular}[c]{@{}c@{}}$\kappa_s$, if $m_s-\mu_s<0$\end{tabular} & $-$ & $\downarrow$ \\ \hline
			$\mu_s$ & $-$ & $\uparrow$ \\ \hline
			$m_s$ & $-$ & $\uparrow$ \\ \hline
			\begin{tabular}[c]{@{}c@{}}$\kappa_{s,p}$, if $m_{s,p}-\mu_{s,p}>0$\end{tabular} & $-$ & $\downarrow$ \\ \hline
			\begin{tabular}[c]{@{}c@{}}$\kappa_{s,p}$, if $m_{s,p}-\mu_{s,p}<0$\end{tabular} & $-$ & $\uparrow$ \\ \hline
			$\mu_{s,p}$ & $-$ & $\downarrow$ \\ \hline
			$m_{s,p}$ & $-$ & $\downarrow$ \\ \hline
		\end{tabular}
		\caption{Table of variation in $\bar{P_s}$ and $C_{sec}/L$ with increase in different system parameters }
	\end{table} 
	
	\begin{appendices} 
		\section{Proof for Theorem 1} \label{minima_lim_cdf}
		We know that $\gamma^K_{min} = min\{\gamma_1,\cdots,\gamma_K \}=-max\{-\gamma_1,\cdots,-\gamma_K \}$. Now, if we derive the asymptotic distribution of the maximum of $K$ i.i.d. RVs $\hat{\gamma}^K_{max}= max\{\hat{\gamma}_1,\cdots,\hat{\gamma}_K \}$ where $\hat{\gamma}_l=-\gamma_l; \ l=1,\cdots , K$ and $\hat{\gamma}_l \sim F_{\hat{\gamma}}(z)=1-F_\gamma(-z)$ then we can also derive the asymptotic distribution of $\gamma^K_{min}$. Now, we invoke the following theorem to derive the limiting distribution of $\hat{\gamma}^K_{max}$. 
		\begin{theorem}  \label{fischer_tippet}
			Fisher-Tippet Theorem, Limit Laws for Maxima: \\
			Let $\mathlarger{z_1,z_2,\cdots,z_K}$ be a sequence of $K$ i.i.d. RVs and $\mathlarger{M_K=}$ max $\mathlarger{\{z_1,z_2,\cdots,z_K\} }$; if $\mathlarger{\exists}$ the constants that obey $\mathlarger{a_K>0}$ and $\mathlarger{b_K\in \mathbb{R}}$ and some non-degenerate CDF $\mathlarger{G_{\beta}}$ so that when $\mathlarger{K \to \infty}$ we have,
			\begin{equation}
			a_K^{-1}\left(M_K-b_K\right) \xrightarrow[]{D} G_{\beta},
			\label{condn_exist}
			\end{equation}
			where $\mathlarger{\xrightarrow[]{D}}$ denotes convergence in distribution. Then the CDF $\mathlarger{G_{\upsilon}}$ is one of the three CDFs:
			\begin{itemize}
				\item [] $Frechet$\cite{gumbel2012statistics}$ \ : \ \Lambda_1 (z) :=  \begin{cases}
				0, & z \leq 0  \\
				exp(-z^{-\upsilon}), & z>0,
				\end{cases} $ 
				\item []  $Reversed \ Weibull$\cite{gumbel2012statistics}$ \ : \ \Lambda_2 (z) :=  \begin{cases}
				exp(-(-z)^{\upsilon}), & z \leq 0, \\
				1, & z>0,
				\end{cases} $
				\item [] $Gumbel$\cite{gumbel2012statistics}$ \ : \ \Lambda_3(z) :=exp(-exp{(-z)}), \ \ \ z \in \mathbb{R}.$
			\end{itemize}
		\end{theorem}
		\begin{proof}
			Please refer to \cite{de2007extreme} for the proof.
		\end{proof}
		To determine the limiting distribution among these three, we first have to define the Maximum Domain of Attraction (\textit{MDA}). 
		\theoremstyle{definition}
		\begin{defn}{Maximum Domain of Attraction \cite{de2007extreme}:} \label{def_mda}
			The CDF $F$ of i.i.d. RVs $z_1,\cdots, z_K$ belongs to the $MDA$ of the extreme value distribution (EVD) $G_\upsilon$, if and only if $\exists$ the constants obeying $a_K>0$ and $b_K\in \mathbb{R}$, so that (\ref{condn_exist}) holds.
		\end{defn}
		\begin{lemma} \label{mda_weibull}
			Let $F$ be a distribution function and $x^*:=sup\{x:F(x)<1\}$. Let us assume that $F^{''}(x)$ exists and $F^{'}(x)$ is positive for all $x$ in some left neighborhood of $x^*$. If 
			\begin{equation}
			\lim\limits_{x \to x^*}\frac{(x^*-x)f(x)}{1-F(x)}=\upsilon; \ \upsilon>0,
			\label{weibull_contn}
			\end{equation}
			then $F(.)$ belongs to the MDA of the reversed Weibull distribution.
		\end{lemma}
		\begin{proof}
			Please refer to \cite{de2007extreme} for the proof.
		\end{proof}
		Now, if we show that the CDF $\mathlarger{F_{\hat{\gamma}}(z)}$ satisfies the relationship in (\ref{weibull_contn}), then from the definition of the $MDA$ of an EVD, we may conclude that there exists $\mathlarger{a_K}$ and $\mathlarger{b_K}$ satisfying (\ref{condn_exist}). A choice for the corresponding constants of the reversed Weibull distribution is given in \cite{de2007extreme} as $\mathlarger{b_K = 0}$ and $\mathlarger{a_K = -F^{-1}_{\hat{\gamma}}(1-K^{-1})}$.
		\begin{theorem}
			The CDF $\mathlarger{F_{\hat{\gamma}}(z)}$ is in the $MDA$ of the reversed Weibull distribution.
		\end{theorem}
		\begin{proof}
			Here we would have to evaluate the following limit : 
			\begin{equation}
			\lim\limits_{z \to 0} \frac{(-z) f_{\hat{\gamma}}(z)}{1-F_{\hat{\gamma}}(z)}.
			\label{limit_eva}
			\end{equation}
			Now, by exploiting the properties of the transformation of RVs, we have $f_{\hat{\gamma}}(z) = f_{{\gamma}}(-z)$ and $F_{\hat{\gamma}}(z)=1-F_{{\gamma}}(-z)$. Thus, (\ref{limit_eva}) can be evaluated as 
			\begin{equation}
			\lim\limits_{z \to 0} \frac{(-z) f_{{\gamma}}(-z)}{F_{{\gamma}}(-z)}.
			\label{limit_eva2}
			\end{equation}
			The pdf $f_{{\gamma}}(z)$ is given by (\ref{pdf1}), where ${K_2 = \frac{\theta^{(m+\sum\limits_{i=1}^{N}\mu_i)}\Gamma\left[\mu+\sum\limits_{i=1}^{N}\mu_i\right]}{\lambda^m \Gamma[\mu] \Gamma\left[\sum\limits_{i=1}^{N}\mu_i\right]\prod\limits_{i=1}^{N}\theta_i^{\mu_i-m_i}\lambda_i^{m_i}}}$;
			
			\begin{equation}
			\begin{aligned}
			f_{\gamma}(z) =  K_2 z^{-(1+\sum\limits_{i=1}^{N}\mu_{i})} \left(1+\frac{\theta}{z \theta_1}\right)^{-\left(\mu+\sum\limits_{i=1}^{N}\mu_{i}\right)}\times _{(1)}^{(1)}E_D^{(2N)}\left[\mu+\sum_{i=1}^{N}\mu_{i},m,\mu_2-m_2,\cdots, \right.\\ \left. \mu_N-m_N,m_1,\cdots,m_N,\mu,\sum_{i=1}^{N}\mu_{i},
			\frac{z\theta_1(\lambda-\theta)}{\lambda(\theta+z\theta_1)},\frac{\theta(\theta_2-\theta_1)}{\theta_2(\theta+z\theta_1)}, \cdots,\frac{\theta(\theta_N-\theta_1)}{\theta_N(\theta+z\theta_1)}, \frac{\theta(\lambda_1-\theta_1)}{\lambda_1(\theta+z\theta_1)}, \right. \\ \left. \cdots, \frac{\theta(\lambda_N-\theta_1)}{\lambda_N(\theta+z\theta_1)}\right].
			\end{aligned} 
			\label{pdf1}
			\end{equation}
			
			Similarly, from \cite[Eqn. (6)]{kumar2017outage}, we have \begin{equation}
			\begin{aligned}
			F_{\gamma}(z) &=1-K_1\left( \frac{z\theta_1}{\theta+z\theta_1}\right)^{\sum\limits_{i=1}^{N}\mu_i+\mu} \times  ^{(1)}_{(1)}E^{(2N+1)}_{D}\left [\sum_{i=1}^{N}\mu_i+\mu,m,1,\mu_2-m_2,\cdots,\mu_N-m_N,m_1, \right. \\  & \left. \cdots,m_N,\mu,1+\sum_{i=1}^{N}\mu_i;  \frac{(\lambda-\theta)z\theta_1}{(\theta+z\theta_1)\lambda},\frac{\theta}{\theta+z\theta_1},\frac{\theta\theta_2-\theta\theta_1}{\theta_2(\theta+z\theta_1)},\cdots,\frac{\theta\theta_N-\theta\theta_1}{\theta_N(\theta+z\theta_1)},\frac{\theta\lambda_1-\theta\theta_1}{\lambda_1(\theta+z\theta_1)}, \right. \\ & \left. \cdots,\frac{\theta\lambda_N-\theta\theta_1}{\lambda_N(\theta+z\theta_1)} \right],
			\end{aligned}
			\label{cdf2}
			\end{equation}
			
			where ${K_1 =\frac{\Gamma\left[ \sum\limits_{1=1}^N \mu_i+\mu\right]\left( \prod\limits_{i=1}^N\theta_i^{-(\mu_i-m_i)}\lambda_i^{-m_i}\right)\theta^{\sum\limits_{i=1}^N\mu_i+m}}{\Gamma\left[ \sum\limits_{1=1}^N \mu_i+1\right]z^{\sum\limits_{i=1}^N\mu_i}\lambda^m\Gamma[\mu]}}$. We now have to evaluate the limit of $F_\gamma(-z)$ in the denominator of (\ref{limit_eva2}), and the above expression of the CDF is available in the $[1-$CCDF] form. For ease of further analysis we reformulate the CDF as follows\footnote{This proof is not included in this paper since it is derived by repeating steps very similar to the derivation of CCDF in \cite{kumar2017outage}.} : 
			
			\begin{align}
			& F_\gamma(z) = \frac{z^{-\sum\limits_{i=1}^N\mu_i}\theta^{\sum\limits_{i=1}^N\mu_i+m} \Gamma \left[\sum\limits_{i=1}^N\mu_i + \mu \right]\lambda^{-m}}{\Gamma[\mu+1]\Gamma\left[\sum\limits_{i=1}^N\mu_i \right] \prod\limits_{i=1}^N \theta_i^{\mu_i-m_i}\lambda_i^{m_i}} \ \left(\frac{\theta_1 z}{\theta+ \theta_1 z} \right)^{\mu+ \sum\limits_{i=1}^N\mu_i}\\ & _{(1)}^{(2)}E_D^{2N+1}\left[\mu+\sum\limits_{i=1}^N\mu_i,1,m,\mu_2-m_2,\cdots,\mu_N-m_N,m_1,\cdots,m_N,\mu+1,\sum\limits_{i=1}^N\mu_i,\frac{\theta_1 z}{\theta+\theta_1 z}, \right. \\ & \left. \frac{\theta(\lambda-\theta)z}{\lambda(\theta+\theta_1 z)}, \frac{\theta(\theta_2-\theta_1)}{\theta_2(\theta+\theta_1z)},\cdots,\frac{\theta(\theta_N-\theta_1)}{\theta_N(\theta+\theta_1z)},\frac{\theta(\lambda_1-\theta_1)}{\lambda_1(\theta+\theta_1z)},\cdots,\frac{\theta(\lambda_N-\theta_1)}{\lambda_N(\theta+\theta_1z)} \right].
			\end{align}
			Now, we can make use of the following properties of the limits to proceed with the evaluation of (\ref{limit_eva2}): 
			\begin{itemize}
				\item $\lim\limits_{x \to a} [f(x)g(x)] = \lim\limits_{x \to a} f(x). \lim\limits_{x \to a}g(x)$
				\item  $\lim\limits_{x \to a} \frac{f(x)}{g(x)} = \frac{\lim\limits_{x \to a} f(x)}{\lim\limits_{x \to a}g(x)}$, if $\lim\limits_{x \to a}g(x) \neq 0$.
			\end{itemize}
			We first consider the ratio without the $E_D(.)$ terms. Here, we have
			\begin{equation}
			\lim\limits_{z \to 0} \frac{\frac{(-z)^{\mu}\theta^{m+\sum\limits_{i=1}^N\mu_i}\Gamma\left[\mu+\sum\limits_{i=1}^N\mu_i \right] \theta_1^{\mu+\sum\limits_{i=1}^N\mu_i}}{\lambda^m \Gamma[\mu]\Gamma[\sum\limits_{i=1}^N\mu_i]\prod\limits_{i=1}^N \theta_i^{\mu_i-m_i}\lambda_i^{m_i}  \left(\theta-z \theta_1 \right)^{\mu+[\sum\limits_{i=1}^N\mu_i}}}{\frac{(-z)^{\mu}\theta^{m+\sum\limits_{i=1}^N\mu_i}\Gamma\left[\mu+\sum\limits_{i=1}^N\mu_i \right] \theta_1^{\mu+\sum\limits_{i=1}^N\mu_i}}{\lambda^m \Gamma[\mu+1]\Gamma[\sum\limits_{i=1}^N\mu_i]\prod\limits_{i=1}^N \theta_i^{\mu_i-m_i}\lambda_i^{m_i}  \left(\theta-z \theta_1 \right)^{\mu+[\sum\limits_{i=1}^N\mu_i}}} = \lim\limits_{z \to 0} \frac{\Gamma[\mu+1]}{\Gamma[\mu]} = \mu.
			\end{equation}
			Now, if we analyze the $E_D(.)$ term in the numerator, we have 
			\begin{equation}
			\begin{aligned}
			& _{(1)}^{(1)}E_D^{(2N)}\left[\mu+\sum_{i=1}^{N}\mu_{i},m,\mu_2-m_2,\cdots,  \mu_N-m_N,m_1,\cdots,m_N,\mu,\sum_{i=1}^{N}\mu_{i},
			\frac{-z\theta_1(\lambda-\theta)}{\lambda(\theta-z\theta_1)},\frac{\theta(\theta_2-\theta_1)}{\theta_2(\theta-z\theta_1)}, \right.\\ & \left. \cdots, \frac{\theta(\theta_N-\theta_1)}{\theta_N(\theta-z\theta_1)}, \frac{\theta(\lambda_1-\theta_1)}{\lambda_1(\theta-z\theta_1)},\cdots, \frac{\theta(\lambda_N-\theta_1)}{\lambda_N(\theta-z\theta_1)}\right],
			\end{aligned} 
			\label{ed_pdf1}
			\end{equation} and the $E_D(.)$ function has the following series expansion : 
			\begin{equation}
			\begin{aligned}
			_{(1)}^{(k)}E_D^{(n)}(a,b_1,\cdots,b_n;c,c';x_1,\cdots,x_n) =   \sum\limits_{i_1\cdots i_n=0}^{\infty} \frac{(a)_{i_1+\cdots+i_n}(b_1)_{(i_1)}\cdots (b_n)_{(i_n)}x_1^{i_1}\cdots x_n^{i_n}}{(c)_{(i_1+\cdots+i_k)}(c')_{(i_{k+1}+\cdots+i_n)}i_1! \cdots i_n!}.
			\end{aligned}
			\end{equation}
			Thus, (\ref{ed_pdf1}) can be expanded as 
			\begin{align}
			\sum\limits_{p_1,\cdots,p_{2N}=0}^\infty &  \frac{\left( \mu+\sum\limits_{i=1}^{N}\mu_{i} \right)_{p_1+\cdots+p_{2N}}(m)_{p_1}\left(\mu_2-m_2 \right)_{p_2}\cdots\left(\mu_N-m_N \right)_{p_N}(m_1)_{p_{N+1}}\cdots(m_N)_{p_{2N}}}{(\mu)_{p_1}\left(\sum_{i=1}^{N}\mu_{i} \right)_{p_2+\cdots+p_{2N}}}  \\ & \times \prod\limits_{i=1}^{2N} \frac{x_i^{p_i}}{p_i!},
			\end{align}
			where $x_1 =\frac{-z\theta_1(\lambda-\theta)}{\lambda(\theta-z\theta_1)} $, $x_i = \frac{\theta(\theta_i-\theta_1)}{\theta_i(\theta-z\theta_1)}; \ i=2\cdots N$ and  $x_i = \frac{\theta(\lambda_i-\theta_1)}{\lambda_i(\theta-z\theta_1)}; \ i=N+1\cdots 2N$. Note that, for $p_1 \neq 0$, we have $\lim\limits_{z \to 0}x_1 = 0$. Hence, at $\lim z \to 0$, only the terms corresponding to $p_1 = 0$ will remain with $x_i;i=2,\cdots,2N$ evaluated at $z \to 0$. Similarly, if we now consider the $E_D(.)$ term in the denominator (from the CDF expression), it has the following series expansion : 
			\begin{align}
			\sum\limits_{p_1,\cdots,p_{2N+1}=0}^\infty &  \frac{\left( \mu+\sum\limits_{i=1}^{N}\mu_{i} \right)_{p_1+\cdots+p_{2N}}(1)_{p_1}(m)_{p_2}\left(\mu_2-m_2 \right)_{p_2}\cdots\left(\mu_N-m_N \right)_{p_N}(m_1)_{p_{N+1}}\cdots(m_N)_{p_{2N}}}{(\mu+1)_{p_1+p_2}\left(\sum_{i=1}^{N}\mu_{i} \right)_{p_3+\cdots+p_{2N}}}  \\ & \times \prod\limits_{i=1}^{2N+1} \frac{x_i^{p_i}}{p_i!},
			\end{align}
			where $x_1 =\frac{-z\theta_1}{(\theta-z\theta_1)} $,$x_2=\frac{\theta(\lambda-\theta)(-z)}{\lambda(\theta-\theta_1 z)}$ $x_i = \frac{\theta(\theta_i-\theta_1)}{\theta_i(\theta-z\theta_1)}; \ i=3\cdots N+1$ and  $x_i = \frac{\theta(\lambda_i-\theta_1)}{\lambda_i(\theta-z\theta_1)}; \ i=N+2\cdots 2N+1$. Note that whenever $p_1 \neq 0$, $p_2 \neq 0$, $\lim\limits_{z \to 0}x_1 = 0$ and $\lim\limits_{z \to 0}x_2 = 0$, respectively. Hence, at $\lim z \to 0$, only the terms corresponding to $p_1 =p_2= 0$ will remain with $x_i;i=3,\cdots,2N+1$ evaluated at $z \to 0$.  Now, note that this set of remaining terms is the same for both the $E_D$ terms in the numerator as well as the denominator. Hence, the ratio of these terms evaluates to one. Thus, we have
			\begin{equation}
			\lim\limits_{z \to 0}\frac{(-z) f_{{\gamma}}(-z)}{F_{{\gamma}}(-z)} = \lim\limits_{z \to 0} \frac{(-z) f_{\hat{\gamma}}(z)}{1-F_{\hat{\gamma}}(z)} = \mu. 
			\label{weibull_contn_test}
			\end{equation} 

		\end{proof}
		Now we know that the asymptotic distribution of $\hat{\gamma}^K_{max}$ is a reversed Weibull distribution, hence we conclude that the asymptotic distribution of the minimum of $K$ SIR RVs ($\gamma^K_{min}$) is a Weibull distribution with shape parameter $\upsilon=\mu$ and the shape parameter $a_K$ as given in (\ref{asymp_cdf}).

		\section{Derivation of rate of convergence} \label{rate_covg}
		
		To prove the result in Theorem \ref{rate_cnvg}, we first define the $\delta$-neighborhood of generalized pareto distribution (GPD) for a Weibull RV. Let the $\delta$-neighbourhood be denoted by $Q_2(\delta)$ and the GPD for a Weibull RV be denoted by $W_{\{2,\nu \}}$. The Extreme Value Distributions (EVDs) lies in the $\delta$ neighbourhood of one of three GPD $W_{\{i,\nu\}}; \  i=1,2,3$ with $\delta=1$ .
		
		\theoremstyle{definition}
		\begin{defn}$\delta$-neighborhood $Q_2(\delta)$ of the GPD  $W_{\{2,\nu\}}$ \cite{falk2010laws} is defined as
			$Q_2(\delta)$ := \{F\ : \ $\omega(F)<$ $\infty$ and $F$ has a density $f$ on $[z_0,\omega(F)]$ for some $z_0<\omega(F)$ such that for some shape parameter $\nu >0$ and some scale parameter $a>0$  on  $[z_0,\omega(F)]$, we have,
			\begin{equation}
			f(z)=\frac{1}{a}W'_{2,\nu}\left(\frac{z-\omega(F)}{a}\right)(1+O((1-W_{2,\nu}(z-\omega(F)))^\delta) \},
			\label{pdf_condition}
			\end{equation}  
			where ${\omega(F):= sup \{z \in \mathbb{R} : F(z)<1 \}}$.
			In fact the GPD for the Weibull distribution is defined in \cite{falk2010laws} as ${W_{2,\nu} = 1 - (-z)^{\nu}; \ -1\leq z \leq 0}$ and using this, (\ref{pdf_condition}) can be rewritten as 
			\begin{equation}
			f(z) = \frac{\nu}{a}\left( \frac{-z+\omega(F)}{a}\right)^{\nu-1}\left( 1+O(((-z+\omega(F))^{\nu})^\delta)\right).
			\label{fz_with_gpd}
			\end{equation}
		\end{defn}
		This definition says that, if a PDF $f$ on  $[z_0,\omega(F)]$ for some $z_0<\omega(F)$ can be written in the form of (\ref{fz_with_gpd}), then the corresponding CDF $F$ belongs to the $\delta$-neighborhood $Q_{2}(\delta)$ of the Weibull distribution\footnote{For a  real or complex valued function $g_1(x)$ and a strictly positive real valued function $g_2(x)$ both defined on some unbounded subset of $\mathbb{R}^+$, we say $g_1(x)=O(g_2(x))$, iff $\exists$ $M\in \mathbb{R}^+$ and $x_0 \in \mathbb{R}$ such that, $|g_1(x)| \leq Mg_2(x) $ $\forall x \geq x_0$.}. The PDF of the RV $\hat{\gamma}=-\gamma$ is given by \cite{subhash2019asymptotic},
		
		\begin{equation}
		\begin{aligned}
		f_{\gamma}(z) =  K_1 (-z)^{-\mu} \left(\frac{\theta_1}{\theta - z \theta_1}\right)^{\left(\mu+\sum\limits_{i=1}^{N}\mu_{i}\right)}\times _{(1)}^{(1)}E_D^{(2N)}\left[\mu+\sum_{i=1}^{N}\mu_{i},m,\mu_2-m_2,\cdots, \right.\\ \left. \mu_N-m_N,m_1,\cdots,m_N;\mu,\sum_{i=1}^{N}\mu_{i};
		\frac{-z\theta_1(\lambda-\theta)}{\lambda(\theta-z\theta_1)},\frac{\theta(\theta_2-\theta_1)}{\theta_2(\theta-z\theta_1)}, \cdots,\frac{\theta(\theta_N-\theta_1)}{\theta_N(\theta-z\theta_1)}, \frac{\theta(\lambda_1-\theta_1)}{\lambda_1(\theta-z\theta_1)}, \right. \\ \left. \cdots, \frac{\theta(\lambda_N-\theta_1)}{\lambda_N(\theta-z\theta_1)}\right],
		\end{aligned} 
		\label{pdf}
		\end{equation}
		where ${K_1 = \frac{\theta^{(m+\sum\limits_{i=1}^{N}\mu_i)}\Gamma\left[\mu+\sum\limits_{i=1}^{N}\mu_i\right]}{\lambda^m \Gamma[\mu] \Gamma\left[\sum\limits_{i=1}^{N}\mu_i\right]\prod\limits_{i=1}^{N}\theta_i^{\mu_i-m_i}\lambda_i^{m_i}}}$. The ${E_D^{(2N)}(.)}$ term in the above expression has the following series expansion from \cite{exton1976multiple}: 
		\begin{equation}
		_{(1)}^{(1)}E_D^{(N)}[a,b_1,\cdots,b_N;c,c';x_1,\cdots,x_N] = \sum\limits_{p_1,\cdots,p_N=0}^{\infty}\frac{(a)_{p_1+\cdots+p_N}\prod\limits_{i=1}^{N}(b_i)_{p_i}\prod\limits_{i=1}^{N}x_i^{p_i}}{(c)_{p_1}(c')_{p_2+\cdots+p_N}p_1!\cdots p_N!}.
		\end{equation}
		Using the above series expansion, we rewrite (\ref{pdf}) as 
		
		\begin{equation}
		\begin{aligned}
		f_{\gamma}(z) = & K_1 (-z)^{-\mu} \left(\frac{\theta_1}{\theta - z \theta_1}\right)^{\left(\mu+\sum\limits_{i=1}^{N}\mu_{i}\right)}\sum\limits_{p_1,\cdots,p_{2N}=0}^{\infty}\frac{\left(\mu+\sum\limits_{i=1}^{N}\mu_i\right)_{p_1+\cdots+p_{2N}}}{(\mu)_{p_1}}  \\ 
		&  \times \frac{(m)_{p_1}  \prod\limits_{i=2}^{N}(\mu_i-m_i)_{p_i}  \prod\limits_{i=N+1}^{2N}(m_i)_{p_i} }{\left(\sum\limits_{i=1}^{N}\mu_i\right)_{p_2+\cdots+p_{2N}}}\prod_{i=1}^{2N}\frac{z_i^{p_i}}{p_i!},
		\end{aligned}
		\label{pdf_series}
		\end{equation}
		where ${z_1 = \frac{(\lambda-\theta)(-z)\theta_1}{\lambda(\theta-z\theta_1)}}$, ${z_i=\frac{\theta(\theta_i-\theta_1)}{\theta_i(\theta-z\theta_1)}}$ for ${i \in \{2,\cdots N\}}$ and ${z_i=\frac{\theta(\lambda_i-\theta_1)}{\lambda_i(\theta-z\theta_1)}}$ for  ${i \in \{N+1,\cdots, 2N\}}.$ We then expand the $2N$ fold summation in (\ref{pdf_series}) into three terms: the first term with all the iterating variables ${p_1, p_2,..., p_{2N}}$ taking the value zero, the second term with exactly one non-zero iterating variable and the third term with the rest. By expanding, (\ref{pdf_series}) becomes the expression given in (\ref{fz_expand}) where $\rho = \mu+\sum\limits_{i=1}^{N}\mu_i$. Now, the term $\left({\frac{\theta}{\theta-z\theta_1}}\right)^{p_j}$ present in ${Term \ b}$ of (\ref{fz_expand}) has the following series expansion: 
		\begin{equation}
		\left({\frac{\theta}{\theta-z\theta_1}}\right)^{p_j} = 1+\frac{p_j \theta_1}{\theta} z + \frac{\theta_1^2\left(p_j+p_j^2 \right)}{2 \theta^2}z^2 + \frac{\left( 2p_j+3p_j^2+p_j^3\right)\theta_1^3}{6\theta^3}z^3 + \mathcal{O}\left(z^4\right).
		\end{equation}
		Similarly, the term $\left( \frac{-z\theta_1}{\theta-z\theta_1}\right)^{p_1}$ has the following series expansion : 
		\begin{equation}
		\left( \frac{-z\theta_1}{\theta-z\theta_1}\right)^{p_1} = \left(\frac{-z \theta_1}{\theta}\right)^{p_1} \left \lbrace 1+ \frac{p_1 \theta_1 z}{\theta} + \frac{\theta_1^2 \left(p_1+p_1^2 \right)z^2}{2\theta^2} + \frac{\theta_1^3 \left(2p_1+3p_1^2+p_1^3 \right)z^3}{6 \theta^3} + \mathcal{O}\left(z^4 \right)  \right \rbrace . 
		\end{equation}
		Thus, Term (a) will have all powers of $z \geq1$ and Term (b) will have all powers of $z \geq0$. Similarly, we can see that Term 3 will have all powers of $z \geq2$. Thus, we can rewrite the pdf expression as follows :
		\begin{equation}
		f_Z(z) = K_1 (-z)^{-\mu} \left(1 + K_2 (-z) + \mathcal{O}(-z)^2 \right)
		\label{pdf_simple}
		\end{equation} where $K_2$ will be a term independent of $z$. Comparing (\ref{pdf_simple}) with (\ref{fz_with_gpd}) and substituting $\omega(F)=0$, we can observe that the pdf of $\hat{\gamma}$ belongs to the domain of attraction of the reversed Weibull distribution with $\nu \times \delta=1$. Thus, we have $\delta={\mu}^{-1}$. 
		\begin{equation}
		\begin{aligned}
		& f_{\gamma}(z) =  \  \left(K_1 (-z)^{-\mu} \left(\frac{\theta_1}{\theta - z \theta_1}\right)^{\rho} \right) \left\lbrace  \underbrace{1}_\textit{Term 1}+  \right. \\ &  \left.  \sum\limits_{p_1=0}^\infty \underbrace{\underbrace{ \frac{(\rho)_{p_1} (m)_{p_1} }{(\mu)_{p_1} p_1!}\left( \frac{(\lambda-\theta)z\theta_1}{ \lambda(z\theta_1 + \theta)}\right)^{p_1}}_\textit{Term a} + \sum\limits_{j=2}^{2N} \sum\limits_{p_j=0}^\infty  \underbrace{\sum\limits_{k=2}^{N} \frac{(\rho)_{p_j}(\mu_k-m_k)_{p_j} }{\left(\sum\limits_{i=1}^{N}\mu_i\right)_{p_j}}\left( z_k\right)^{p_j}+\sum\limits_{k=N+1}^{2N}\frac{(\rho)_{p_j}(m_k)_{p_j} }{\left(\sum\limits_{i=1}^{N}\mu_i\right)_{p_j}}\left( z_k\right)^{p_j}}_\textit{Term b} \ }_\textit{Term 2}    \right. \\ & \left.   + \underbrace{  \sum\limits_{\underset{{ \exists \ i_1,i_2 s.t \ p_{i_1}p_{i_2}\neq0 \ \forall \ i_1\neq i_2}}{p_1,\cdots,p_{2N}=0;}}^{\infty} \frac{(\rho)_{p_1+\cdots+p_{2N}}}{(\mu)_{p_1}}\frac{(m)_{p_1}  \prod\limits_{i=2}^{N}(\mu_i-m_i)_{p_i}  \prod\limits_{i=N+1}^{2N}(m_i)_{p_i} \prod\limits_{i=1}^{2N}\frac{z_i^{p_i}}{p_i!} }{\left(\sum\limits_{i=1}^{N}\mu_i\right)_{p_2+\cdots+p_{2N}}}}_\textit{Term 3} \right\rbrace .
		\end{aligned}
		\label{fz_expand}
		\end{equation}
		
		Now that we have identified the $\delta$ neighbourhood for ${F_{\gamma}(z)}$, we make use of the following lemma from \cite{falk2010laws} to conclude the proof. 
		
		\begin{lemma}
			Suppose that the CDF $F$ (of i.i.d. RVs $z_1,\cdots,z_K$) is in the $\delta$ neighborhood $Q_2(\delta)$ of the GPD $W_{2,\nu}$  then there obviously exist constant $a > 0$ such that  $f(z)=\frac{1}{a}W'_{2,\nu}\left(\frac{z-\omega(F)}{a}\right)(1+O((1-W_{2,\nu}(z-\omega(F)))^\delta) \}$ ${W_{2,\nu}(z))^\delta)} $ for all $z$ in the left neighborhood of ${\omega(W_{2,\nu})}$. Consequently we have, 
			
			\begin{equation}
			\sup_{B\in \mathbb{B}} \left| \mathbb{P}\left(\left(\left( \frac{M_K}{a}  \right) / K^{\nu}\right) \in B\right) - G_{\nu}(B)\right| = O\left(\left(\frac{1}{K}\right)^{\delta} + \frac{1}{K}\right),
			\end{equation}
			where $\mathbb{B}$ denotes the Borel $\sigma$ algebra on $\mathbb{R}$ and $M_K = max \{z_1,\cdots,z_K \} $.	
		\end{lemma}
		Since the CDF $F_{\gamma}(z)$ belongs to the $\delta$ neighborhood of $Q_2(\delta)$, by the previous lemma, the rate of convergence is $O\left(\left(\frac{1}{K}\right)^{\delta} + \frac{1}{K}\right)$ with $\delta = \mu^{-1}$ .
		
		\section{Proof for Theorem (\ref{evt_rate_main}) } \label{rate_convg}
		To prove this, we first utilize the continuous mapping theorem, which is given as follows \cite{billingsley2013convergence}: 
		\begin{theorem} \label{cts_map}
			Let $\{X_n\}_{n=1}^\infty$ be a sequence of random variables and $X$ another random variable, all assuming values from the same metric space $\mathcal{X}$. Let $\mathcal{Y}$ be another metric space, $f:\mathcal{X} \to \mathcal{Y}$ be a measurable function and $C_f := \{x : \ f \ is \ continuous \ at \ x\}$. Assuming that $X_n \xrightarrow[\text{}]{\text{D}}X$ and $\mathbb{P}(X \in C_f)=1$, we have $f(X_n) \xrightarrow[\text{}]{\text{D}}f(X)$.
		\end{theorem}
		Let $R^K_{min}=\text{log}_2(1+ \gamma_{min}^K)$.  Since $f(x)=log_2(1+x)$ is a continuous function, using Theorem \ref{cts_map}, $R^K_{min} \xrightarrow[\text{}]{\text{D}}{R}_{min}$. Finally, we use the monotonic convergence theorem, which is given below \cite{billingsley2008probability}.
		\begin{theorem} \label{mct}
			Let $g_n \geq 0$ be a sequence of measurable functions such that $g_n(\omega) \to g(\omega) \ \forall \ \omega$ except maybe on a set of measure zero and $g_n(\omega) \geq g_{n+1}(\omega), \ n \geq 1$. We then have
			\begin{equation}
			\lim\limits_{n \to \infty} \int g_n \ d\mu = \int g \ d\mu.
			\end{equation}
		\end{theorem}
		Here, we know that $\gamma_{min}^K \geq \gamma_{min}^{K+1}, \ \forall \ K$ and hence $\mathbb{P}(\gamma_{min}^K \leq \omega) \leq \mathbb{P}(\gamma_{min}^{K+1} \leq \omega)$ and thus $1-F_{\gamma_{min}^K}(\omega) \geq 1-F_{\gamma_{min}^{K+1}}(\omega)$. The logarithmic function is monotonic and hence  $1-F_{R^K_{min}}(\omega) \geq 1-F_{R^{K+1}_{min}}(\omega)$. For a positive RV $X$, note that the expectation is given by 
		\begin{equation}
		\mathbb{E}[X] = \int\limits_{0}^{\infty} \mathbb{P}(X>x) \ dx = \int\limits_{0}^{\infty} (1-F_X(x)) \ dx.
		\end{equation} 
		Thus, by making use of Theorem \ref{mct} we have $\lim\limits_{K \to \infty} \mathbb{E}[R^K_{min}] = \lim\limits_{K \to \infty} \int\limits_{0}^{\infty}  \mathbb{P}(R^K_{min}>\omega) \ d\omega= \int\limits_{0}^{\infty} \lim\limits_{K \to \infty} \mathbb{P}(R^K_{min}>\omega) \ d\omega = \mathbb{E}[{R}_{min}] $. Thus, we have the required result.
		
	\end{appendices}
	\bibliographystyle{IEEEtran}
	\bibliography{reference}

\begin{thebibliography}{10}
\providecommand{\url}[1]{#1}
\csname url@samestyle\endcsname
\providecommand{\newblock}{\relax}
\providecommand{\bibinfo}[2]{#2}
\providecommand{\BIBentrySTDinterwordspacing}{\spaceskip=0pt\relax}
\providecommand{\BIBentryALTinterwordstretchfactor}{4}
\providecommand{\BIBentryALTinterwordspacing}{\spaceskip=\fontdimen2\font plus
\BIBentryALTinterwordstretchfactor\fontdimen3\font minus
  \fontdimen4\font\relax}
\providecommand{\BIBforeignlanguage}[2]{{%
\expandafter\ifx\csname l@#1\endcsname\relax
\typeout{** WARNING: IEEEtran.bst: No hyphenation pattern has been}%
\typeout{** loaded for the language `#1'. Using the pattern for}%
\typeout{** the default language instead.}%
\else
\language=\csname l@#1\endcsname
\fi
#2}}
\providecommand{\BIBdecl}{\relax}
\BIBdecl

\bibitem{dalta2009cr}
D.~{Datla}, A.~M. {Wyglinski}, and G.~J. {Minden}, ``{A spectrum surveying
  framework for dynamic spectrum access networks},'' \emph{IEEE Transactions on
  Vehicular Technology}, vol.~58, no.~8, pp. 4158--4168, Oct 2009.

\bibitem{wang2011cr}
J.~{Wang}, M.~{Ghosh}, and K.~{Challapali}, ``{Emerging cognitive radio
  applications: A survey},'' \emph{IEEE Communications Magazine}, vol.~49,
  no.~3, pp. 74--81, March 2011.

\bibitem{liang2011cr}
Y.~{Liang}, K.~{Chen}, G.~Y. {Li}, and P.~{Mahonen}, ``{Cognitive radio
  networking and communications: an overview},'' \emph{IEEE Transactions on
  Vehicular Technology}, vol.~60, no.~7, pp. 3386--3407, Sep. 2011.

\bibitem{goldsmith2009cr}
A.~{Goldsmith}, S.~A. {Jafar}, I.~{Maric}, and S.~{Srinivasa}, ``{Breaking
  spectrum gridlock with cognitive radios: an information theoretic
  perspective},'' \emph{Proceedings of the IEEE}, vol.~97, no.~5, pp. 894--914,
  May 2009.

\bibitem{miridakis2018mimo}
N.~I. Miridakis, T.~A. Tsiftsis, and G.~C. Alexandropoulos, ``Mimo underlay
  cognitive radio: Optimized power allocation, effective number of transmit
  antennas and harvest-transmit tradeoff,'' \emph{IEEE Transactions on Green
  Communications and Networking}, vol.~2, no.~4, pp. 1101--1114, 2018.

\bibitem{zhao2007cr}
Q.~{Zhao} and B.~M. {Sadler}, ``{A survey of dynamic spectrum access},''
  \emph{IEEE Signal Processing Magazine}, vol.~24, no.~3, pp. 79--89, May 2007.

\bibitem{khoshkholgh2010cr}
M.~G. {Khoshkholgh}, K.~{Navaie}, and H.~{Yanikomeroglu}, ``{Access strategies
  for spectrum sharing in fading environment: overlay, underlay, and mixed},''
  \emph{IEEE Transactions on Mobile Computing}, vol.~9, no.~12, pp. 1780--1793,
  Dec 2010.

\bibitem{patel2017achievable}
A.~Patel, M.~Z.~A. Khan, S.~N. Merchant, U.~B. Desai, and L.~Hanzo, ``{The
  achievable rate of interweave cognitive radio in the face of sensing
  errors},'' \emph{IEEE Access}, vol.~5, pp. 8579--8605, 2017.

\bibitem{patel2018many}
A.~Patel, M.~Z.~A. Khan, S.~Merchant, U.~Desai, and L.~Hanzo, ``{How many
  cognitive channels should the primary user share?}'' \emph{IEEE Wireless
  Commun.}, no.~99, pp. 1--8, 2018.

\bibitem{zou2015relay}
Y.~Zou, B.~Champagne, W.-P. Zhu, and L.~Hanzo, ``{Relay-selection improves the
  security-reliability trade-off in cognitive radio systems},'' \emph{IEEE
  Trans. Commun.}, vol.~63, no.~1, pp. 215--228, 2015.

\bibitem{ding2019security}
X.~Ding, Y.~Zou, G.~Zhang, X.~Chen, X.~Wang, and L.~Hanzo, ``{The
  security-reliability tradeoff of multiuser scheduling aided energy harvesting
  cognitive radio networks},'' \emph{IEEE Transactions on Communications},
  2019.

\bibitem{zhang2009cr}
R.~{Zhang}, ``{On peak versus average interference power constraints for
  protecting primary users in cognitive radio networks},'' \emph{IEEE
  Transactions on Wireless Communications}, vol.~8, no.~4, pp. 2112--2120,
  April 2009.

\bibitem{suraweera2010cr}
H.~A. {Suraweera}, P.~J. {Smith}, and M.~{Shafi}, ``{Capacity limits and
  performance analysis of cognitive radio with imperfect channel knowledge},''
  \emph{IEEE Transactions on Vehicular Technology}, vol.~59, no.~4, pp.
  1811--1822, May 2010.

\bibitem{rezki2012cr}
Z.~{Rezki} and M.~{Alouini}, ``{Ergodic capacity of cognitive radio under
  imperfect channel-state information},'' \emph{IEEE Transactions on Vehicular
  Technology}, vol.~61, no.~5, pp. 2108--2119, Jun 2012.

\bibitem{kang2011cr}
X.~{Kang}, R.~{Zhang}, Y.~{Liang}, and H.~K. {Garg}, ``{Optimal power
  allocation strategies for fading cognitive radio channels with primary user
  outage constraint},'' \emph{IEEE Journal on Selected Areas in
  Communications}, vol.~29, no.~2, pp. 374--383, February 2011.

\bibitem{smith2013cr}
P.~J. {Smith}, P.~A. {Dmochowski}, H.~A. {Suraweera}, and M.~{Shafi}, ``{The
  effects of limited channel knowledge on cognitive radio system capacity},''
  \emph{IEEE Transactions on Vehicular Technology}, vol.~62, no.~2, pp.
  927--933, Feb 2013.

\bibitem{hanif2017cr}
M.~{Hanif}, H.~{Yang}, and M.~{Alouini}, ``{Transmit antenna selection for
  power adaptive underlay cognitive radio With instantaneous interference
  constraint},'' \emph{IEEE Transactions on Communications}, vol.~65, no.~6,
  pp. 2357--2367, June 2017.

\bibitem{patel2016achievable}
A.~Patel, M.~Z.~A. Khan, S.~Merchant, U.~B. Desai, and L.~Hanzo, ``{Achievable
  rates of underlay-based cognitive radio operating under rate limitation},''
  \emph{IEEE Trans. Veh. Tech.}, vol.~65, no.~9, pp. 7149--7159, 2016.

\bibitem{louis_ergodic}
L.~Sibomana and H.-J. Zepernick, ``{Ergodic capacity of multiuser scheduling in
  cognitive radio networks: analysis and comparison},'' \emph{Wireless
  Communications and Mobile Computing}, vol.~16, no.~16, pp. 2759--2774, 2016.

\bibitem{badarneh_asymptotic}
Y.~H. Al-Badarneh, C.~N. Georghiades, and M.~Alouini, ``{Asymptotic performance
  analysis of generalized user selection for interference-limited multiuser
  secondary networks},'' \emph{IEEE Transactions on Cognitive Communications
  and Networking}, pp. 1--1, 2019.

\bibitem{yacoub_k_mu}
M.~Yacoub, ``{The {\(\kappa\)}-{\(\mu\)} distribution and the
  {\(\eta\)}-{\(\mu\)} distribution},'' \emph{IEEE Antennas and Propagat.
  Mag.}, vol.~49, no.~1, pp. 68--81, Feb 2007.

\bibitem{paris2014statistical}
J.~F. Paris, ``{Statistical characterization of $\kappa-\mu$ shadowed
  fading},'' \emph{IEEE Trans. Veh. Tech.}, vol.~63, no.~2, pp. 518--526, 2014.

\bibitem{cotton_d2d}
S.~L. Cotton, ``{Human body shadowing in cellular device-to-device
  communications: channel modeling using the shadowed {\(\kappa\)}-{\(\mu\)}
  fading model},'' \emph{IEEE Journal of Sel. Topics in Comm.}, vol.~33, no.~1,
  pp. 111--119, Jan 2015.

\bibitem{pozas_shadowed}
L.~Moreno-Pozas, F.~J. Lopez-Martinez, J.~F. Paris, and E.~Martos-Naya, ``{The
  {\(\kappa\)}-{\(\mu\)} shadowed fading model: unifying the
  {\(\kappa\)}-{\(\mu\)} and {\(\eta\)}-{\(\mu\)} distributions},'' \emph{IEEE
  Trans. Veh. Tech.}, vol.~65, no.~12, pp. 9630--9641, Dec 2016.

\bibitem{celia2014capacity}
C.~García-Corrales, F.~J. Cañete, and J.~F. Paris, ``{Capacity of
  $\kappa-\mu$ shadowed fading channels},'' \emph{International Journal of
  Antennas and Propagation}, 2014.

\bibitem{zhang2015effective}
J.~Zhang, L.~Dai, W.~H. Gerstacker, and Z.~Wang, ``{Effective capacity of
  communication systems over $\kappa-\mu$ shadowed fading channels},''
  \emph{Electron. Lett.}, vol.~51, no.~19, pp. 1540--1542, 2015.

\bibitem{chen2016outage}
C.~Chen, M.~Shu, Y.~Wang, and C.~Zhang, ``{Outage probability analysis for MRC
  in {\(\kappa\)}-{\(\mu\)} shadowed fading channels with co-channel
  interference},'' in \emph{IEEE Int. Conf. on Info. and Autom.}, Aug 2016, pp.
  270--275.

\bibitem{li2017rate}
X.~Li, J.~Li, L.~Li, J.~Jin, J.~Zhang, and D.~Zhang, ``{Effective rate of MISO
  systems over $\kappa $ - $\mu $ shadowed fading channels},'' \emph{IEEE
  Access}, vol.~5, pp. 10\,605--10\,611, 2017.

\bibitem{zhang2017hos}
J.~Zhang, X.~Chen, K.~P. Peppas, X.~Li, and Y.~Liu, ``{On high-order capacity
  statistics of spectrum aggregation systems over {\(\kappa\)}-{\(\mu\)} and
  $\kappa $ - $\mu $ shadowed fading channels},'' \emph{IEEE Trans. on Comm.},
  vol.~65, no.~2, pp. 935--944, Feb 2017.

\bibitem{chandrasekaran2015performance}
G.~Chandrasekaran and S.~Kalyani, ``{Performance analysis of cooperative
  spectrum sensing over $\kappa-\mu$ shadowed fading},'' \emph{IEEE Wireless
  Commun. Lett.}, vol.~4, no.~5, pp. 553--556, 2015.

\bibitem{thomas2016error}
V.~A. Thomas, S.~Kumar, S.~Kalyani, M.~El-Hajjar, K.~Giridhar, and L.~Hanzo,
  ``{Error vector magnitude analysis of fading {SIMO} channels relying on {MRC}
  reception},'' \emph{IEEE Trans. Commun.}, vol.~64, no.~4, pp. 1786--1797,
  2016.

\bibitem{morales2012outage}
D.~Morales-Jimenez, J.~F. Paris, and A.~Lozano, ``{Outage probability analysis
  for MRC in {\(\eta\)}-{\(\mu\)} fading channels with co-channel
  interference},'' \emph{IEEE Commun. Lett.}, vol.~16, no.~5, pp. 674--677, May
  2012.

\bibitem{paris2013outage}
J.~F. Paris, ``{Outage Probability in {\(\eta\)}-{\(\mu\)}/{\(\eta\)}-{\(\mu\)}
  and {\(\kappa\)}-{\(\mu\)}/ {\(\eta\)}-{\(\mu\)} interference-limited
  Scenarios},'' \emph{IEEE Trans. Commun.}, vol.~61, no.~1, pp. 335--343,
  January 2013.

\bibitem{ermolova2014outage}
N.~Y. Ermolova and O.~Tirkkonen, ``{Outage probability analysis in generalized
  fading channels with co-channel interference and background noise:
  {\(\eta\)}-{\(\mu\)}/{\(\eta\)}-{\(\mu\)}, {\(\eta\)}-{\(\mu\)}/
  {\(\kappa\)}-{\(\mu\)}, and {\(\kappa\)}-{\(\mu\)}/ {\(\eta\)}-{\(\mu\)}
  scenarios},'' \emph{IEEE Trans. Wireless Commun.}, vol.~13, no.~1, pp.
  291--297, January 2014.

\bibitem{kumar2015coverage}
S.~Kumar and S.~Kalyani, ``{Coverage probability and rate for
  {\(\kappa\)}-{\(\mu\)}/ {\(\eta\)}-{\(\mu\)} fading channels in
  interference-limited scenarios},'' \emph{IEEE Trans. Wireless Commun.},
  vol.~14, no.~11, pp. 6082--6096, Nov 2015.

\bibitem{kumar2015outage}
S.~Kumar, G.~Chandrasekaran, and S.~Kalyani, ``{Analysis of outage probability
  and capacity for {\(\kappa\)}-{\(\mu\)}/ {\(\eta\)}-{\(\mu\)} faded
  channel},'' \emph{IEEE Commun. Lett.}, vol.~19, no.~2, pp. 211--214, Feb
  2015.

\bibitem{zhang2017performance}
J.~Zhang, X.~Li, I.~S. Ansari, Y.~Liu, and K.~A. Qaraqe, ``{Performance
  analysis of dual-hop DF satellite relaying over $\kappa-\mu$ shadowed fading
  channels},'' in \emph{Proc. IEEE Wireless Commun. \& Netw. Conf.}\hskip 1em
  plus 0.5em minus 0.4em\relax IEEE, 2017, pp. 1--6.

\bibitem{parthasarathy2017coverage}
S.~Parthasarathy and R.~K. Ganti, ``{Coverage analysis in downlink poisson
  cellular network with {\(\kappa\)}-{\(\mu\)} shadowed fading},'' \emph{IEEE
  Wireless Commun. Lett.}, vol.~6, no.~1, pp. 10--13, Feb 2017.

\bibitem{parthasarathy2018evm}
S.~Parthasarathy, S.~Kumar, R.~K. Ganti, S.~Kalyani, and K.~Giridhar, ``{Error
  vector magnitude analysis in generalized fading with co-channel
  interference},'' \emph{IEEE Trans. Commun.}, vol.~66, no.~1, pp. 345--354,
  Jan 2018.

\bibitem{srinivasan2018secrecy}
M.~{Srinivasan} and S.~{Kalyani}, ``{Secrecy capacity of $\kappa-\mu$ shadowed
  fading channels},'' \emph{IEEE Communications Letters}, vol.~22, no.~8, pp.
  1728--1731, Aug 2018.

\bibitem{kumar2017outage}
S.~Kumar and S.~Kalyani, ``{Outage probability and rate for $\kappa -\mu $
  shadowed fading in interference limited scenario},'' \emph{IEEE Trans.
  Wireless Commun.}, vol.~16, no.~12, pp. 8289--8304, 2017.

\bibitem{aghazadeh2018performance}
B.~Aghazadeh and M.~Torabi, ``{Performance evaluation of multi-user diversity
  in a SIMO spectrum sharing system with reduced CSI load},'' \emph{Digital
  Signal Processing}, vol.~72, pp. 160--170, 2018.

\bibitem{khan2015performance}
F.~A. Khan, K.~Tourki, M.-S. Alouini, and K.~A. Qaraqe, ``{Performance analysis
  of an opportunistic multi-user cognitive network with multiple primary
  users},'' \emph{Wireless Communications and Mobile Computing}, vol.~15,
  no.~16, pp. 2004--2019, 2015.

\bibitem{jindal2006multicast}
N.~Jindal and Z.~Q. Luo, ``{Capacity limits of multiple antenna multicast},''
  in \emph{Proc. IEEE Int. Symposium Inform. Theory}, July 2006, pp.
  1841--1845.

\bibitem{park2008multicast}
S.~Y. Park and D.~J. Love, ``{Capacity limits of multiple antenna multicasting
  using antenna subset selection},'' \emph{IEEE Trans. Signal Process.},
  vol.~56, no.~6, pp. 2524--2534, June 2008.

\bibitem{park2009multicast}
------, ``{Outage performance of multi-antenna multicasting for wireless
  networks},'' \emph{IEEE Trans. Wireless Commun.}, vol.~8, no.~4, pp.
  1996--2005, April 2009.

\bibitem{oyman2007scheduling}
O.~Oyman, ``{Opportunism in multiuser relay channels: scheduling, routing and
  spectrum reuse},'' in \emph{Proc. IEEE Int. Symposium Inform. Theory}, June
  2007, pp. 286--290.

\bibitem{oyman2008scheduling}
O.~Oyman and M.~Z. Win, ``{Power-bandwidth tradeoff in multiuser relay channels
  with opportunistic scheduling},'' in \emph{Proc. Allerton Conf. Commun.
  Control Comput.}, Sept 2008, pp. 72--78.

\bibitem{oyman2010scheduling}
O.~Oyman, ``{Opportunistic scheduling and spectrum reuse in relay-based
  cellular networks},'' \emph{IEEE Trans. Wireless Commun.}, vol.~9, no.~3, pp.
  1074--1085, March 2010.

\bibitem{ahmadi2012scheduling}
S.~Al-Ahmadi, ``{The asymptotic capacity of opportunistic scheduling over
  shadowed Nakagami fading channels},'' in \emph{Wireless Commun. Mobile
  Comput.}, Aug 2012.

\bibitem{kountouris2009scheduling}
M.~Kountouris and J.~G. Andrews, ``{Throughput scaling laws for wireless ad-hoc
  networks with relay selection},'' in \emph{Proc. IEEE Veh. Tech. Conf.},
  April 2009, pp. 1--5.

\bibitem{xue2010mi}
Q.~Xue and G.~Abreu, ``{Mutual information of amplify-and-forward relaying with
  partial relay selection},'' in \emph{2010 5th International ICST Conference
  on Communications and Networking in China}, Aug 2010, pp. 1--4.

\bibitem{xia2014scheduling}
M.~Xia and S.~Aissa, ``{Spectrum-sharing multi-hop cooperative relaying:
  performance analysis using extreme value theory},'' \emph{IEEE Trans.
  Wireless Commun.}, vol.~13, no.~1, pp. 234--245, January 2014.

\bibitem{biswas2016relay}
S.~Biswas, S.~Vuppala, J.~Xue, and T.~Ratnarajah, ``{On the performance of
  relay aided millimeter wave networks},'' \emph{IEEE J. Sel. Topics Signal
  Process.}, vol.~10, no.~3, pp. 576--588, April 2016.

\bibitem{xu2016relay}
H.~Xu, L.~Sun, P.~Ren, Q.~Du, and Y.~Wang, ``{Cooperative privacy preserving
  scheme for downlink transmission in multiuser relay networks},'' \emph{IEEE
  Trans. Inf. Forensics Security}, vol.~12, no.~4, pp. 825--839, 2017.

\bibitem{kalyani2012gamma}
S.~Kalyani and R.~M. Karthik, ``{The asymptotic distribution of maxima of
  independent and identically distributed sums of correlated or non-identical
  gamma random variables and its applications},'' \emph{IEEE Trans. Commun.},
  vol.~60, no.~9, pp. 2747--2758, September 2012.

\bibitem{pun2011mimo}
M.~O. Pun, V.~Koivunen, and H.~V. Poor, ``{Performance analysis of joint
  opportunistic scheduling and receiver design for MIMO-SDMA downlink
  systems},'' \emph{IEEE Trans. Commun.}, vol.~59, no.~1, pp. 268--280, January
  2011.

\bibitem{gao2018massive}
Y.~Gao, H.~Vinck, and T.~Kaiser, ``{Massive MIMO antenna selection: switching
  architectures, capacity bounds, and optimal antenna selection algorithms},''
  \emph{IEEE Trans. Signal Process.}, vol.~66, no.~5, pp. 1346--1360, 2018.

\bibitem{subhash2019asymptotic}
A.~Subhash, M.~Srinivasan, and S.~Kalyani, ``Asymptotic maximum order statistic
  for sir in $\kappa$--$\mu$ shadowed fading,'' \emph{IEEE Transactions on
  Communications}, 2019.

\bibitem{gumbel2012statistics}
E.~J. Gumbel, \emph{{Statistics of extremes}}.\hskip 1em plus 0.5em minus
  0.4em\relax Courier Corporation, 2012.

\bibitem{ban2009multi}
T.~W. Ban, W.~Choi, B.~C. Jung, and D.~K. Sung, ``Multi-user diversity in a
  spectrum sharing system,'' \emph{IEEE Transactions on Wireless
  Communications}, vol.~8, no.~1, pp. 102--106, 2009.

\bibitem{chen2006unified}
C.-J. Chen and L.-C. Wang, ``A unified capacity analysis for wireless systems
  with joint multiuser scheduling and antenna diversity in nakagami fading
  channels,'' \emph{IEEE Transactions on Communications}, vol.~54, no.~3, pp.
  469--478, 2006.

\bibitem{song2006asymptotic}
G.~Song and Y.~Li, ``Asymptotic throughput analysis for channel-aware
  scheduling,'' \emph{IEEE Transactions on Communications}, vol.~54, no.~10,
  pp. 1827--1834, 2006.

\bibitem{kumar2019errata}
S.~Kumar and S.~Kalyani, ``{Errata to the paper “Outage Probability and Rate
  for $\kappa$--$\mu$ Shadowed Fading in Interference Limited Scenario”},''
  \emph{IEEE Trans. Wireless Commun.}, 2019.

\bibitem{lauricode}
``Code for evaluating fnd,'' \url{http://faculty.smu.edu/rbutler/}, {Accessed:
  2018 [Online]}.

\bibitem{dubey1970compound}
S.~D. Dubey, ``{Compound gamma, beta and F distributions},'' \emph{Metrika},
  vol.~16, no.~1, pp. 27--31, 1970.

\bibitem{tepedelenlioglu2011applications}
C.~Tepedelenlioglu, A.~Rajan, and Y.~Zhang, ``Applications of stochastic
  ordering to wireless communications,'' \emph{IEEE Transactions on Wireless
  Communications}, vol.~10, no.~12, pp. 4249--4257, 2011.

\bibitem{dhillon2013downlink}
H.~S. Dhillon, M.~Kountouris, and J.~G. Andrews, ``Downlink mimo hetnets:
  Modeling, ordering results and performance analysis,'' \emph{IEEE
  Transactions on Wireless Communications}, vol.~12, no.~10, pp. 5208--5222,
  2013.

\bibitem{madhusudhanan2012stochastic}
P.~Madhusudhanan, J.~G. Restrepo, Y.~E. Liu, T.~X. Brown, and K.~R. Baker,
  ``Stochastic ordering based carrier-to-interference ratio analysis for the
  shotgun cellular systems,'' \emph{IEEE Wireless Communications Letters},
  vol.~1, no.~6, pp. 565--568, 2012.

\bibitem{ji2012capacity}
J.~Ji and W.~Chen, ``{Capacity analysis of multicast transmission schemes in a
  spectrum-sharing scenario},'' \emph{IET Communications}, vol.~6, no.~17, pp.
  2974--2979, 2012.

\bibitem{shaked}
M.~Shaked and J.~Shanthikumar, \emph{{Stochastic Orders}}, ser. Springer Series
  in Statistics.\hskip 1em plus 0.5em minus 0.4em\relax Springer New York,
  2007.

\bibitem{itu2008requirements}
ITU-R, ``{Requirements related to technical performance for IMT-Advanced radio
  interface (s)},'' \emph{International Telecommunications Union}, 2008.

\bibitem{sesia2011lte}
S.~Sesia, M.~Baker, and I.~Toufik, \emph{{LTE-the UMTS long term evolution:
  from theory to practice}}.\hskip 1em plus 0.5em minus 0.4em\relax John Wiley
  \& Sons, 2011.

\bibitem{de2007extreme}
L.~De~Haan and A.~Ferreira, \emph{{Extreme value theory: an
  introduction}}.\hskip 1em plus 0.5em minus 0.4em\relax Springer Science \&
  Business Media, 2007.

\bibitem{falk2010laws}
M.~Falk, J.~H{\"u}sler, and R.-D. Reiss, \emph{Laws of small numbers: extremes
  and rare events}.\hskip 1em plus 0.5em minus 0.4em\relax Springer Science \&
  Business Media, 2010.

\bibitem{exton1976multiple}
H.~Exton, ``{Multiple hypergeometric functions and applications},'' 1976.

\bibitem{billingsley2013convergence}
P.~Billingsley, \emph{{Convergence of probability measures}}.\hskip 1em plus
  0.5em minus 0.4em\relax John Wiley \& Sons, 2013.

\bibitem{billingsley2008probability}
------, \emph{{Probability and measure}}.\hskip 1em plus 0.5em minus
  0.4em\relax John Wiley \& Sons, 2008.

\end{thebibliography}
\end{document}